\documentclass[12pt]{article}

\usepackage[utf8]{inputenc} % set input encoding (not needed with XeLaTeX)

\usepackage[margin=2cm]{geometry}

\usepackage{graphicx} 

\usepackage{cite}
\usepackage{lineno}

\usepackage{booktabs} % for much better looking tables
\usepackage{array} % for better arrays (eg matrices) in maths
\usepackage{paralist} % very flexible & customisable lists (eg. enumerate/itemize, etc.)
\usepackage{verbatim} % adds environment for commenting out blocks of text & for better verbatim
\usepackage{subfig} % make it possible to include more than one captioned figure/table in a single float

\usepackage{hyperref}
\usepackage{amsthm}
\usepackage{amssymb}
\usepackage{authblk}

\newtheorem{theorem}{Theorem}
\newtheorem{mydef}{Definition}

\title{General superposition states associated to the rotational and inversion symmetries in the phase space.}
\author[1,2]{Julio A. L\'opez-Sald\'ivar}
\affil[1]{Instituto de Ciencias Nucleares, Universidad Nacional Aut\'onoma de M\'exico,
Apdo. Postal 70-543, Ciudad de México 04510, M\'exico}
\affil[2]{Moscow Institute of Physics and Technology, Institutskii per. 9,
Dolgoprudnyi, Moscow~Region 141700, Russia}
\date{julio.lopez.8303@gmail.com}

\begin{document}

\maketitle

\begin{abstract}
The general quantum superposition states containing the irreducible representation of the $n$-dimensional groups associated to the rotational symmetry of the $n$-sided regular polygon i.e., the cyclic group ($C_n$) and the rotational and inversion symmetries of the polygon, i.e., the dihedral group ($D_n$) are defined and studied. It is shown that the resulting states form an $n$-dimensional orthogonal set of states which can lead to the finite representation of specific systems. The correspondence between the symmetric states and the renormalized states, resulting from the selective erasure of photon numbers from an arbitrary, noninvariant initial state, is also established. As an example, the general cyclic Gaussian states are presented. The presence of nonclassical properties in these states as subpoissonian photon statistics is addressed. Also, their use in the calculation of physical quantities as the entanglement in a bipartite system is discussed.
\end{abstract}

\section{Introduction}

The study of symmetries in physics has helped to the simplification of difficult problems. For example, the symmetries in the Hamiltonian dynamical evolution of a quantum system can be related to the definition of different conservation laws which, as in the classical theory, can be used to answer different questions. The use of symmetries in quantum mechanics, in particular the definition of states associated to point symmetry groups has been covered in several works \cite{manko1,manko2,manko3,castanos}. Especially, the states carrying the symmetry of the cyclic group $C_2=\mathbb{Z}/(2 \mathbb{Z})$, also called odd an even cat states, have been of great interest in the past decades. The nonclassical properties of this kind of states have been discussed in \cite{buzek}, together with their use in fundamental quantum theory \cite{sanders,wenger,jeong,stob,wineland} and in the quantum information framework \cite{vanerick,jeong2,ralph,gilchrist,bergmann}.

For several years, there was an impossibility to construct a cat state with a large photon number. Instead of that, the low photon cat states, known as kitten states, were generated \cite{ourjoumtsev1}. After that, the possibility to obtain full cat states has been demonstrated in several studies as: by using the reflexion of a coherent pulse from a optical cavity with one atom \cite{hacker,wang}, the use of homodyne detection in a photon number state \cite{ourjoumtsev2}, the photon subtraction from a squeezed vacuum state in a parametric amplifier \cite{neegaard}, via ancilla-assisted photon subtraction \cite{takahashi}, and by the subtraction of an specific photon number in a squeezed vacuum state \cite{gerrits}. The superposition of coherent states have non-classical features like squeezing of the quadrature components \cite{buzek,janszky,domokos}. There exist a possible experimental implementation of these superpositions \cite{szabo}, in particular superpositions of coherent states on a circle \cite{janszky,domokos,gonzalez}. The states adapted to this type of symmetry have also a connection to the phase-time operators in the harmonic oscillator \cite{susskind,nieto1,pegg1,pegg2}. The definition of states carrying the circle symmetry has been extended by the use of spin coherent states as in \cite{calixto}, also in \cite{calixto1} the use of $su(1,1)$ coherent states on the hyperboloid were considered.

More recently, a proposed method to generate states with higher discrete symmetries, as the ones defined here, has been obtained by the dynamic evolution of a matter-field interaction described by the Tavis-Cummings model \cite{cordero1,cordero2}. There is experimental evidence for the generation of superpositions of four coherent states with a number of 111 photons \cite{vlastakis}. Also, the cluster structure of light nuclei as $^{12}C$ and $^{13}C$ have been describe by the point symmetry groups, as the ones discussed here, $D_{3h}$ and $D_{3h}^\prime$, respectively.

In this work, the generalization of the quantum states associated to the irreducible representations of the group whose elements are the symmetry rotations of the $n$-sided regular polygon, also named the cyclic group ($C_n=\mathbb{Z}/(n \mathbb{Z})$), and the group containing the rotational and reflexion symmetries of the regular polygon, i.e., the dihedral group ($D_n$), is presented. Some of these type of states have been previously defined using coherent states \cite{manko1,manko2,manko3,castanos}. 

In the present work, it is shown that the cyclic and dihedral states form an orthogonal set of states, which can be used to define a discrete representation of states made of the superposition of rotations, in the case of the cyclic group, and rotations plus reflections in the case of the dihedral group. Also it can be seen that this discrete representation can simplify the calculation of quantum parameters as the entanglement between two subsystems within a system. For these reasons, we consider that given the applications of the cyclic and dihedral coherent states in quantum information, the generalization of such states to the noncoherent case is important. 

The proposed method discussed here, makes use of an initial state $\vert \phi \rangle$ which is not invariant under rotations. To define the cyclic states, the superposition of the rotated states $\vert \phi_r \rangle=\hat{R}(\theta_r)\vert \phi \rangle$ ($r=1,\ldots,n$; $\theta_r=2\pi(r-1)/n$), and the characters associated to the $\lambda$-th  irreducible representation and the $r$-th element of the group ($\chi^{(\lambda)}(g_r)$ ), are used. It is also discussed the relation between the cyclic states and the renormalized states obtained from the erasure of certain photon numbers in the photon statistics of $\vert \phi \rangle$ or $\hat{\rho}$, e.g., the cat states associated to the cyclic group $C_2$: $\vert \xi_\pm \rangle = N_\pm (\vert \alpha \rangle\pm\vert-\alpha\rangle)$ are the renormalized states resulting of eliminating the even and odd photon number states from the coherent state $\vert \alpha \rangle$ respectively. 

On the other hand, the dihedral group $D_n$ is the non-Abelian group that contains the rotations and inversions which leave the $n$-sided regular polygon invariant. The elements of the dihedral group are $D_n:\{\hat{R}(\theta_j), \hat{U}_j,\, j=1,\ldots,n \}$, with $\theta_j=2\pi(j-1)/n$, where the inversion operators in the phase space are defined by a rotation plus the complex conjugation ($\hat{C}$), i.e., $\hat{U}_j=\hat{C}\hat{R}(\theta_j)$.

Additionally to pure, non-pure cyclic and dihedral states can be defined through a density matrix. These states correspond to a quantum map of an noninvariant, arbitrary operator $\hat{\rho}$. This type of quantum maps have been recently relevant in quantum information theory. In particular, the quantum maps have been important for the quantum error correction as some of the studied qubit maps represent the interaction between a qubit and an environment \cite{terhal,caruso}. Furthermore, the study of the erasure map, presented here, can be important to figure out the experimental realization of the defined states, as the resulting states, depend on the absorption (erasure) of certain state numbers.

As a remainder of some group characteristics we establish that given a $n$ dimensional group $\{g_r;\, r=1, \ldots,n \}$, a conjugacy class is formed by all the elements $g_k$ which satisfy the similarity transformation $g_k^{-1} g_j g_k=g_j$, where $g_j$ is also a member of the group. An irreducible representation $\lambda$ is the representation of a group that cannot decompose further. To obtain the irreducible representation sometimes the following procedure should be applied: if there exist a similarity transformation of an element of the group $g_j$ which diagonalize it, i.e., $C^{-1}g_j C=A_D$, where $A_D$ is made of diagonal matrices $A_{D_j^{(\lambda)}}$, then the matrices $A_{D_j^{(\lambda)}}$ form an irreducible representation of $g_j$. The character $\chi$ associated to the irreducible representation $\lambda$, is defined as the trace of the diagonal matrix $A_{D_\lambda}$, that is $\chi^{(\lambda)}(g_j)={\rm Tr}\left(A_{D_j^{(\lambda)}}\right)$. Also, all the members of a conjugacy class share the same characters. In the case of the cyclic states the character associated to the irreducible representation $\lambda$ and element $g_r$ of the group is given by $\chi^{(\lambda)}_n(g_r)=e^{2\pi i(\lambda-1)(r-1)/n}$

This work is organized as follows: In section 2 a review of the cyclic states constructed by means of coherent states are presented. The generalization of these type of states for a non-coherent system is then described in section 3. The correspondence between the generalized cyclic state and a renormalized state obtained through the elimination of certain photon numbers in an original system is studied in section 4. In section 5, some examples are given, the cyclic Gaussian states are defined and some of their properties are exemplified. Also, the circle symmetry states are presented as an extension to the states associated to $C_n$, where $n\rightarrow \infty$. In section 6, the idea of the pure cyclic states of $C_n$ is extended to the case of non-pure density matrices. This is done by the definition of a map of the density matrix, which can also be related to the erasure and renormalization of certain photon numbers in the initial state. The usefulness of this kind of systems for the study of the entanglement in a two-mode system is shown in section 7. The dihedral states are defined in section 8. Finally, some conclusions are given.

\section{Cyclic coherent states}
In previous works, different states associated to the irreducible representation of cyclic groups \cite{manko1,manko2,manko3,castanos} have been defined using coherent states \cite{glauber,titulaer,birula,stoler}. The resulting states called crystallized cat states have some interesting properties as subpoissonian photon statistics, squeezing, and antibunching \cite{sun1,sun2,castanos}. Also, it has been demonstrated that they can be generated by the interaction of an atom with an electromagnetic field \cite{hacker,wang}. Here, we present a summary of the definition and some properties of the coherent cyclic states. 

The cyclic group $C_n$ have as elements the discrete rotations associated to the symmetries of the regular polygon of $n$ sides, i.e. $C_n=\{R(\theta_j), \theta_j = 2\pi (j-1) /n, \ {\rm with}\ (j=1,\ldots,n)\}$.  The number of elements is equal to the cycle of the group and they can be divided in different conjugacy classes $\{g_r\}$. The characteristic (or character) of the class $g_r$ for the irreducible representation $\lambda$ is denoted as $\chi^{(\lambda)}_n(g_r)$ is given by the trace of the irreducible representation. It is known that in the case of the cyclic group each element forms its own class ($g_j=R(\theta_j)$) and that the character of the class are the $n$ roots of the identity,
\begin{equation}
\chi^{(\lambda)}_n(g_r)=\exp \left[ \frac{2 i \pi (\lambda-1)(r-1) }{n}\right] \, , \quad {\rm with}\ \lambda,r=1,\ldots,n \, .
\label{chi}
\end{equation}
Additionally, the characters for any two irreducible representations $\lambda$ and $\lambda'$ are orthonormal, i.e.,
\begin{equation}
\frac{1}{n}\sum_{r=1}^n \chi^{(\lambda)}_n(g_r) \chi^{*(\lambda')}_n(g_r)= \delta_{\lambda \lambda'}
\label{ort1}
\end{equation}
and also the sum of the characters over all the irreducible representations $\lambda$ satisfy that
\begin{equation}
\frac{1}{n}\sum_{\lambda=1}^n \chi^{(\lambda)}_n(g_r) \chi^{*(\lambda)}_n(g_{r'})= \delta_{r r'} \, .
\label{ort2}
\end{equation}
These two orthogonality conditions can be quickly checked using the rule for the sum of the identity roots
\begin{equation}
\sum_{j=1}^n \mu_n^j=0\, , \quad {\rm where} \ \mu_n=\exp\left(\frac{2\pi i}{n}\right),
\label{powers}
\end{equation}
such property also leads to the following theorem.
\begin{theorem}
\label{tt1}
Let $r$ be an integer and $\mu_n=\exp(2\pi i/n )$, then $\sum_{j=1}^n \mu_n^{jr}=n\, \delta_{{\rm mod}(r,n),0}$.
\end{theorem}
\begin{proof}
It is clear that for $r$ being a multiple  of $n$: ${\rm mod}(r,n)=0$, $\mu_n^{rj}=1$ and thus the sum $\sum_{j=1}^n \mu_n^{jr}$ is equal to $n$. For $r$ not being a multiple of $n$ (${\rm mod}(r,n) \neq 0$) we remember that the sum 
\[
\sum_{j=1}^n x^j=x\frac{x^n-1}{x-1} \, ,
\]
which in the case of $x=\mu_n^r$, implies
\[
\sum_{j=1}^n \mu_n^{jr}=\mu_n^r \frac{\mu_n^{rn}-1}{\mu_n^r-1}=0\, ,  
\]
as $\mu_n^{rn}=1$. It is important to notice that this property is satisfied for any integer, in particular by $r$ being a negative integer.
\end{proof}

Given the orthogonality properties in Eqs.~(\ref{ort1}) and (\ref{ort2}) one can define a macroscopic quantum state for each one of the irreducible representations of the cyclic group as follows
\begin{equation}
\left\vert \psi^{(\lambda)}_n \right\rangle=\mathcal{N}_\lambda \sum_{r=1}^n \chi^{(\lambda)}_n (g_r) \vert \alpha_r \rangle \, , \quad \sum_{r,r'=1}^n \chi^{(\lambda)}_n (g_r) \chi^{*(\lambda)}_n (g_{r'}) \langle \alpha_{r'} \vert \alpha_r \rangle = \mathcal{N_\lambda}^{-2} \, ,
\label{cats}
\end{equation}
where the coherent state parameter $\alpha_r={\rm Re}(\alpha_r)+i\,{\rm Im}(\alpha_r)$ is given by the rotation of a fixed number $\alpha$ in the complex plane,
\[
\left( \begin{array}{cc} {\rm Re}(\alpha_r) \\  {\rm Im}(\alpha_r)\end{array}\right)= R(\theta_r) \left( \begin{array}{cc} {\rm Re}(\alpha) \\  {\rm Im}(\alpha)\end{array}\right) \, .
\]
It is important to notice that all the states for different irreducible representations form an orthonomal set with $\left\langle \psi_n^{(\lambda)} \Big\vert \psi_n^{(\lambda')} \right\rangle = \delta_{\lambda \lambda'}$. In the case of the cyclic group $C_2$ we have as the result the standard odd and even cat states $\vert \psi^{(1,2)} \rangle= \mathcal{N}_\pm (\vert \alpha \rangle \pm \vert - \alpha \rangle)$, which can have subpoissonian photon statistic, squeezing, and antibunching \cite{castanos}.

The coherent cyclic states $\vert \psi_n^{(\lambda)}\rangle$ are eigenvalues of the power of the annihilation operator $\hat{a}^n$, i.e.,
\[
\hat{a}^n \vert \psi_n^{(\lambda)}\rangle= \alpha^n \vert \psi_n^{(\lambda)}\rangle \, .
\]
Also, one can change the irreducible representation of the state by acting the annihilation operator $\hat{a}$ to another state: 
\[
\hat{a}\vert \psi_n^{(\lambda)}\rangle=\alpha \frac{\mathcal{N}_{\lambda}}{ \mathcal{N}_{\lambda'}} \vert \psi_n^{(\lambda')}\rangle \, ,
\]
where the value of the new irreducible representation depends on the original one $\lambda'(\lambda)$.

\section{Generalization of cyclic states as superpositions of rotations in the phase space.}

The necessity of a generalization of the cyclic states to a superposition of arbitrary, non-coherent systems can be explained by their possible use in quantum information theory. Also, the cyclic states form an orthogonal set of states which can lead to a finite representation of certain quantum systems.

First, let us suppose an initial quantum state $\vert \phi \rangle$ and its representation in the Fock basis
\[
\vert \phi \rangle = \sum_{m=0}^\infty A_m(\phi) \vert m \rangle \, , \quad {\rm with} \ \sum_{m=0}^\infty \vert A_m(\phi) \vert^2 =1 \, .
\]
The discrete rotations in the phase space associated to the symmetries of the regular polygon in the cyclic group $C_n$ are given by the operator $\hat{R}(\theta_j)=\exp(-i \theta_j \hat{n})$, where $\theta_j=2 \pi(j-1)/n$; $j=1, \ldots , n$, and $\hat{n}$ is the bosonic number operator. To every one of the elements of the cyclic group we have then a rotation of the general state $\vert \phi \rangle$, which can be expressed as
\[
\vert \phi_j \rangle = \hat{R}(\theta_j) \vert \phi \rangle \, .
\]
\begin{mydef}
Let $\vert \phi \rangle=\sum_{m=0}^\infty A_m (\phi) \vert m \rangle$ be a quantum state with at least one mean quadrature component ($\hat{x}=(\hat{a}+\hat{a}^\dagger)/\sqrt{2}$, $\hat{p}=i(\hat{a}^\dagger-\hat{a})/\sqrt{2}$) different from zero, i.e., $\langle \phi \vert \hat{x} \vert \phi \rangle \neq 0$, or $\langle \phi \vert \hat{p} \vert \phi \rangle \neq 0$.We define the general cyclic state for the irreducible representation $\lambda$ of the group $C_n$ as
\begin{equation}
\left\vert \psi_n^{(\lambda)} (\phi) \right\rangle = \mathcal{N}_\lambda \sum_{r=1}^n \chi^{(\lambda)}_n (g_r) \vert \phi_r \rangle \, ,
\label{ccy}
\end{equation}
where $\chi_n^{(\lambda)}(g_r)$ is the character associated to the irreducible representation $\lambda$ and to the element of the group $g_r \in C_n$, and where 
\[
\mathcal{N}_\lambda^{-2}=\sum_{r,r'=1}^n  \chi^{(\lambda)}(g_r) \chi^{*(\lambda)}(g_{r'}) \langle \phi_{r'} \vert \phi_r \rangle \, .
\]
\label{defi1}
\end{mydef}

To obtain a well defined state we emphasize that the original state cannot be invariant under the rotations discussed above, i.e., $\vert \phi_r \rangle \neq \vert \phi \rangle$ for $r=2,\ldots,n$. This property can be satisfied when the Wigner function of the state  in the phase space $W(x,p)$ is given by a non symmetric distribution or when the state is not centered at the origin of the phase space, i.e., $\int dx \, dp \, x \, W(x,p)\neq 0$, or $\int dx \, dp \, p \, W(x,p)\neq 0$.

As these states carry the irreducible representation of the group $C_n$, they are invariant, up to a phase, under the discrete rotations $\hat{R}(\theta_j)$. To prove this property, lets suppose the action of the rotation $\hat{R}(\theta_l)$, $1\leq l \leq n$, over the state $\left\vert \psi_n^{(\lambda)} (\phi) \right\rangle$
\[
\hat{R}(\theta_l) \left\vert \psi_n^{(\lambda)} (\phi) \right\rangle= \mathcal{N}_\lambda \sum_{r=1}^n \chi_n^{(\lambda)} (g_r) \hat{R}(\theta_{r+l}) \vert \phi \rangle \, ,
\]
as the character of the representation $\lambda$ is
\[
\chi_n^{(\lambda)} (g_r)=\mu_n^{(\lambda-1)(r-1)}=\mu_n^{(\lambda-1)(r+l-1)}\mu_n^{(1-\lambda)l} = \chi_n^{(\lambda)} (g_{r+l}) \mu_n^{(1-\lambda)l} \, ,
\]
then we obtain
\[
\hat{R}(\theta_l) \left\vert \psi_n^{(\lambda)} (\phi) \right\rangle= \mathcal{N}_\lambda \,  \mu_n^{(1-\lambda)l} \sum_{r=1}^n \chi_n^{(\lambda)} (g_{r+l}) \hat{R}(\theta_{r+l}) \vert \phi \rangle \, ,
\]
given the periodicity of the characters and the rotation operators ($\mu_n^{x+n}=\mu_n^x$, $\hat{R}(\theta_{j+n})=\hat{R}(\theta_{j})$), this sum give us, up to a phase, the same state as the original, i.e.,
\begin{equation}
\hat{R}(\theta_l) \left\vert \psi_n^{(\lambda)} (\phi) \right\rangle=   \mu_n^{(1-\lambda)l} \left\vert \psi_n^{(\lambda)}(\phi) \right\rangle \, .
\label{cyc_inv}
\end{equation}

It can also be seen that by the use of the explicit form of the rotated states in the Fock basis $\vert \phi_r \rangle=\sum_{m=0}^\infty A_m(\phi) e^{-i \theta_r m} \vert m \rangle$, one obtains
\[
\left\vert \psi_n^{(\lambda)} (\phi) \right\rangle = \mathcal{N}_\lambda \sum_{r=1}^n \sum_{m=0}^\infty \chi^{(\lambda)}_n (g_r) A_m(\phi) e^{-i\theta_r m}\vert m \rangle \, ,
\]
which can be also rewritten as
\begin{equation}
\left\vert \psi_n^{(\lambda)} (\phi) \right\rangle = \mathcal{N}_\lambda \sum_{r=1}^n \sum_{m=0}^\infty  \mu_n^{(\lambda-1-m)(r-1)} A_m(\phi) \vert m \rangle \, .
\label{cyc1}
\end{equation}
Given the characteristics of the sum of the powers of the parameter $\mu_n$, expressed in Eq.~(\ref{powers}), one can show that the different states for the cyclic group $C_n$ form an ortonormal set. To show this, lets suppose the inner product of two cyclic states with irreducible representations $\lambda$, and $\lambda^\prime$, i.e.,
\[
\left\langle \psi_n^{(\lambda^\prime)} (\phi) \right\vert \psi_n^{(\lambda)} (\phi) \Big\rangle= \mathcal{N}_\lambda  \mathcal{N}_{\lambda^\prime}  \sum_{r,r'=1}^n \sum_{m,m'=0}^\infty A_m(\phi) A_{m'}^*(\phi)\, \mu_n^{(\lambda-1-m)(r-1)} \mu_n^{(1-\lambda^\prime+m')(r'-1)} \delta_{m',m} \, ,
\]
performing first the sums over the parameter $r'$, we have
\[
\left\langle \psi_n^{(\lambda^\prime)} (\phi) \right\vert \psi_n^{(\lambda)} (\phi) \Big\rangle= \mathcal{N}_\lambda  \mathcal{N}_{\lambda^\prime}  n \sum_{m=0}^\infty \sum_{r=1}^n \vert A_m (\phi) \vert^2 \mu_n^{\lambda^\prime-1-m}\, \mu_n^{(\lambda-1-m)(r-1)} \delta_{{\rm mod}(1-\lambda^\prime+m,n),0} \, .
\]
As established by Theorem~\ref{tt1}, this sum is different from zero when $1-\lambda^\prime+m=s n$ (with $s\in \mathbb{Z}$ ). This leads to the condition $m=sn-1+\lambda^\prime$. From this, we can change the sum over $m$ to a sum over $s$, obtaining
\[
\left\langle \psi_n^{(\lambda^\prime)} (\phi) \right\vert \psi_n^{(\lambda)} (\phi) \Big\rangle= \mathcal{N}_\lambda  \mathcal{N}_{\lambda^\prime}  \sum_{s=0}^\infty \sum_{r=1}^n \vert A_{ns-1+\lambda} (\phi) \vert^2 \mu_n^{(\lambda-\lambda^\prime)r} \mu_n^{\lambda-\lambda'}\, .
\]
Similarly to the previous step, the sum over the parameter $r$ is different from zero when $\lambda-\lambda'=s' n$ with $s' \in \mathbb{Z}$. As the parameters satisfy $1\leq \lambda,\lambda' \leq n$, the only possible value  is that $\lambda-\lambda'=0$, so
\[
\left\langle \psi_n^{(\lambda^\prime)} (\phi) \right\vert \psi_n^{(\lambda)} (\phi) \Big\rangle= \mathcal{N}_\lambda  \mathcal{N}_{\lambda^\prime}  \sum_{s=0}^\infty \vert A_{ns-1+\lambda} (\phi) \vert^2 \delta_{\lambda,\lambda'} \, ,
\]
which in the case $\lambda\neq \lambda'$ is equal to zero and by the expression for the normalization constant in Def.~(\ref{defi1}) is equal to one when $\lambda= \lambda'$. Finally, arriving to the expression
\[
\left\langle \psi_n^{(\lambda^\prime)} (\phi) \right\vert \psi_n^{(\lambda)} (\phi) \Big\rangle=\delta_{\lambda,\lambda'} \, .
\]

 Other important properties of the cyclic states are addressed in the next section.

\section{State erasure as a quantum map and the cyclic states.}
In this section, the connection between the erasure map and the cyclic states is studied. This correspondence can lead to the experimental implementation of the cyclic states as these states can be seen as coming from the absorption (or erasure) of certain photon numbers.

The general cyclic states defined above, can also be defined as the result of selective loss of information in a quantum system, that is, from the erasure of a subset of states of an original state $\vert \phi \rangle=\sum_{n=0}^\infty A_m(\phi) \vert m \rangle$.

As an example, one can suppose the selective erasure of the probabilities $A_m(\phi)$ for all values of odd $m$, and after this erasure, the renormalization of the state is performed. In that case, one will have the following state made with only even number states
\begin{equation}
\vert \psi_{even} \rangle = N \sum_{m \ even} A_m (\phi) \vert m \rangle \, ,
\label{even}
\end{equation}
where $N$ is the normalization constant $N^{-2}=\sum_{m \ even} \vert A_m(\phi) \vert^2$. Lets compare the previous expression with the cyclic state for $n=2$, $\lambda=1$: $\left\vert \psi_2^{(1)} (\phi)\right\rangle$. This state is given by
\[
\left\vert \psi_2^{(1)} (\phi) \right\rangle=\mathcal{N}_1 \sum_{r=1}^2 \sum_{m=0}^\infty \mu_2^{r m} A_m (\phi) \vert m \rangle, \quad \mu_2=-1\, .
\]
By performing the sum over $r$, we then obtain
\[
\left\vert \psi_2^{(1)} (\phi)\right\rangle=\mathcal{N}_1  \sum_{m=0}^\infty (1+(-1)^m) A_m (\phi) \vert m \rangle=\mathcal{N}_1 \sum_{m \ even} 2  A_m (\phi) \vert m \rangle \, ,
\]
which is the same expression as Eq.~(\ref{even}) with $N=2\mathcal{N}_1$. The same can be done to the state resulting of the elimination of even states, which is equal to the cyclic state with $n=2$, $\lambda=2$, i.e., $\left\vert \psi_2^{(2)}(\phi) \right\rangle=N \sum_{m \ odd} A_m (\phi) \vert m \rangle$. In general, the equality between the cyclic states and the states resulting of the elimination of certain number states can be established in the following theorem.

\begin{theorem}
Let  $n$ and $\lambda$ be two positive integers with $\lambda\leq n$, and $\vert \Psi_{n,\lambda} (\phi) \rangle$ be the renormalized state obtained after the elimination of the number states $\vert m \rangle$ in $\vert \phi \rangle=\sum_{m=0}^\infty A_m (\phi) \vert m \rangle$, which do not satisfy the condition ${\rm mod}(m-\lambda+1,n)=0$, then $\vert \Psi_{n,\lambda} (\phi)\rangle$ is equal to the cyclic state $\left\vert \psi_n^{(\lambda)}(\phi)\right\rangle$ up to a phase.
\label{teo2}
\end{theorem}
\begin{proof}
The state after the erasure map $\vert \Psi_{n,\lambda} (\phi) \rangle$, has the following expression
\[
\vert \Psi_{n,\lambda} \rangle=N_{\lambda,n} \sum_{m} A_m (\phi) \delta_{{\rm mod}(m-\lambda+1,n),0} \vert m \rangle
\]
which only contains the number states that satisfy $m-\lambda+1=ln$ (with $l$ an nonnegative integer), then 
\begin{equation}
\vert \Psi_{n,\lambda} (\phi) \rangle=N_{\lambda,n} \sum_{l} A_{\lambda-1+ln} (\phi)  \vert \lambda-1+ln \rangle \, .
\label{elim}
\end{equation}
On the other hand, by using the property for the roots of the identity ($\mu_n$) given in Theorem~\ref{tt1} ($\sum_{j=1}^n \mu_n^{j l}=n \,  \delta_{{\rm mod}(l,n),0}$), in the definition of $\left\vert \psi_n^{(\lambda)}(\phi)\right\rangle$ in Eq.~(\ref{cyc1}), we can show that
\[
 \sum_{r=1}^n  \mu_n^{(\lambda-1-m)(r-1)}=n \, \mu_n^{1-\lambda} \delta_{{\rm mod}(\lambda-1-m,n),0} \, ,
\]
this means that only the states with $m=\lambda-1+l n$ (with $l$ a nonnegative integer) are part of $\left\vert \psi_n^{(\lambda)}\right\rangle$, i. e.,
\begin{equation}
\left\vert \psi_n^{(\lambda)}(\phi)\right\rangle=n \, \mathcal{N}_\lambda \, \mu_n^{1-\lambda} \sum_{l=0}^\infty  A_{\lambda-1+l n} (\phi)\vert \lambda-1+l n \rangle \, .
\label{cyc2}
\end{equation}
Finally, when comparing Eqs.~(\ref{elim}) and~(\ref{cyc2}) we arrive to the conclusion
\begin{equation}
\vert \Psi_{n, \lambda} \rangle = \mu_n^{1-\lambda} \left\vert \psi_n^{(\lambda)} (\phi)\right\rangle \, ,
\end{equation}
with the relation between the normalization constants being $n\, \mathcal{N}_\lambda=N_{\lambda,n}$, and the phase between the cyclic state and the erasure state, being $\mu_n^{1-\lambda}=\exp{(2 \pi i (1-\lambda)/n)}$.
\end{proof}
Given this identification, it can be seen that the photon number statistics for the state $\left\vert \psi_n^{(\lambda)} (\phi) \right\rangle$ contain only the photon numbers which satisfy ${\rm mod}(m-\lambda+1,n)=0$.

The correspondence between the cyclic states and the states resulting from the quantum erasure map can lead to the experimental realization of the cyclic states. One can for example think of an initial nonivariant state $\vert \phi \rangle$, with a small mean photon number ($\langle \phi \vert \hat{n} \vert \phi \rangle \approx 0$). If one has a process where the number states $\vert 1 \rangle$ or $\vert 2 \rangle$ are erased, e.g., by the absorption of one or two photons of the electromagnetic field, then one can expect that the resulting state will be similar to a cyclic state.

\section{Examples.}

\subsection{Cyclic Gaussian states.}

Here we define different superpositions of Gaussian states associated to the cyclic groups. These superpositions are connected with the squeezed states defined in \cite{nieto,hillery}. As an example of the general procedure described above, one can define cyclic states using Gaussian wavepackets as initial systems.
Suppose a general one dimensional Gaussian state in the position basis
\begin{equation}
\psi (x)= \left(\frac{a+a^*}{\pi}\frac{1+2a}{1+2a^*}\right)^{1/4} \exp\left\{ -\frac{b^2+b b^*}{4(a+a^*)}\right\} \exp \left\{ -a x^2+b x\right\}\, , \quad a_R>0\, , \ b\neq0 \, ,
\label{gaus}
\end{equation}
with $a=a_R+ia_I$, $b=b_R+i b_I$. This state can be characterized by the mean values of the quadrature components $(\hat{p},\hat{q})$, and the corresponding covariance matrix $\sigma$. Which in the case of the state (\ref{gaus}) are
\begin{equation}
\langle \hat{x} \rangle=\frac{b+b^*}{2(a+a^*)}, \quad \langle \hat{p} \rangle= \frac{i (a b^*-a^* b)}{a+a^*}\, , \quad \sigma=\frac{1}{2(a+a^*)}\left(\begin{array}{cc}
4 \vert a \vert^2 & i(a-a^*) \\ i(a-a^*) & 1
\end{array}\right) \, .
\end{equation}
When this state is rotated in the phase space using the propagator $\langle x \vert \hat{R}(\theta) \vert y \rangle$, where $\hat{R}(\theta)=\exp(-i\theta \hat{n})$ is the rotation operator, the obtained state is still Gaussian with new parameters $a(\theta)$, $b (\theta)$ given in terms of the original Gaussian parameters $a$, and $b$, as follows
\begin{eqnarray}
a (\theta)&=&\frac{2 i a \cos \theta-\sin \theta}{2 (i \cos \theta-2  a \sin \theta)} \, , \nonumber \\
b (\theta)&=& \frac{b}{\cos \theta+2 i a \sin \theta} \, .
\end{eqnarray}
The cyclic Gaussian state for the irreducible representation $\lambda$ of the group $C_n$ is then given by the expression
\begin{equation}
\Psi_n^{(\lambda)}(x)= \mathcal{N}_\lambda \sum_{r=1}^n \chi_n^{(\lambda)} (g_r) \psi_r (x)\, ,
\label{gcic}
\end{equation}
with a value of $\psi_r (x)$ analogous to the initial state of Eq.~(\ref{gaus})
\begin{eqnarray}
\psi_r (x)=\left(\frac{a(\theta_r)+a^*(\theta_r)}{\pi}\frac{1+2a(\theta_r)}{1+2a^*(\theta_r)}\right)^{1/4} \exp\left\{ -\frac{b^2(\theta_r)+b(\theta_r) b^*(\theta_r)}{4(a(\theta_r)+a^*(\theta_r))}\right\} \times \nonumber \\
\exp \{ -a(\theta_r) x^2+b (\theta_r) x \} \, .
\end{eqnarray}
Given this expression one can construct then the cyclic states using Eq.~(\ref{gcic}).

For the cyclic group $C_2$, the cyclic states can be described by the following two orthogonal states
\begin{equation}
\Psi^{(1,2)}_2 (x)=\mathcal{N}_{1,2} \, e^{-a x^2}(e^{b x} \pm e^{-b x}) \, , \quad \mathcal{N}_{1,2}=\left( \frac{a+a^*}{\pi}\frac{1+2a}{1+2a^*}\right)^{1/4} \frac{e^{-\frac{b^2+b b^*}{4(a+a^*)}}}{\sqrt{2}\left(1\pm e^{-\frac{b b^*}{a+a^*}}\right)^{1/2}} \, ,
\label{goga}
\end{equation}
%Figure 1
\begin{figure}
\centering
(a)\includegraphics[scale=0.25]{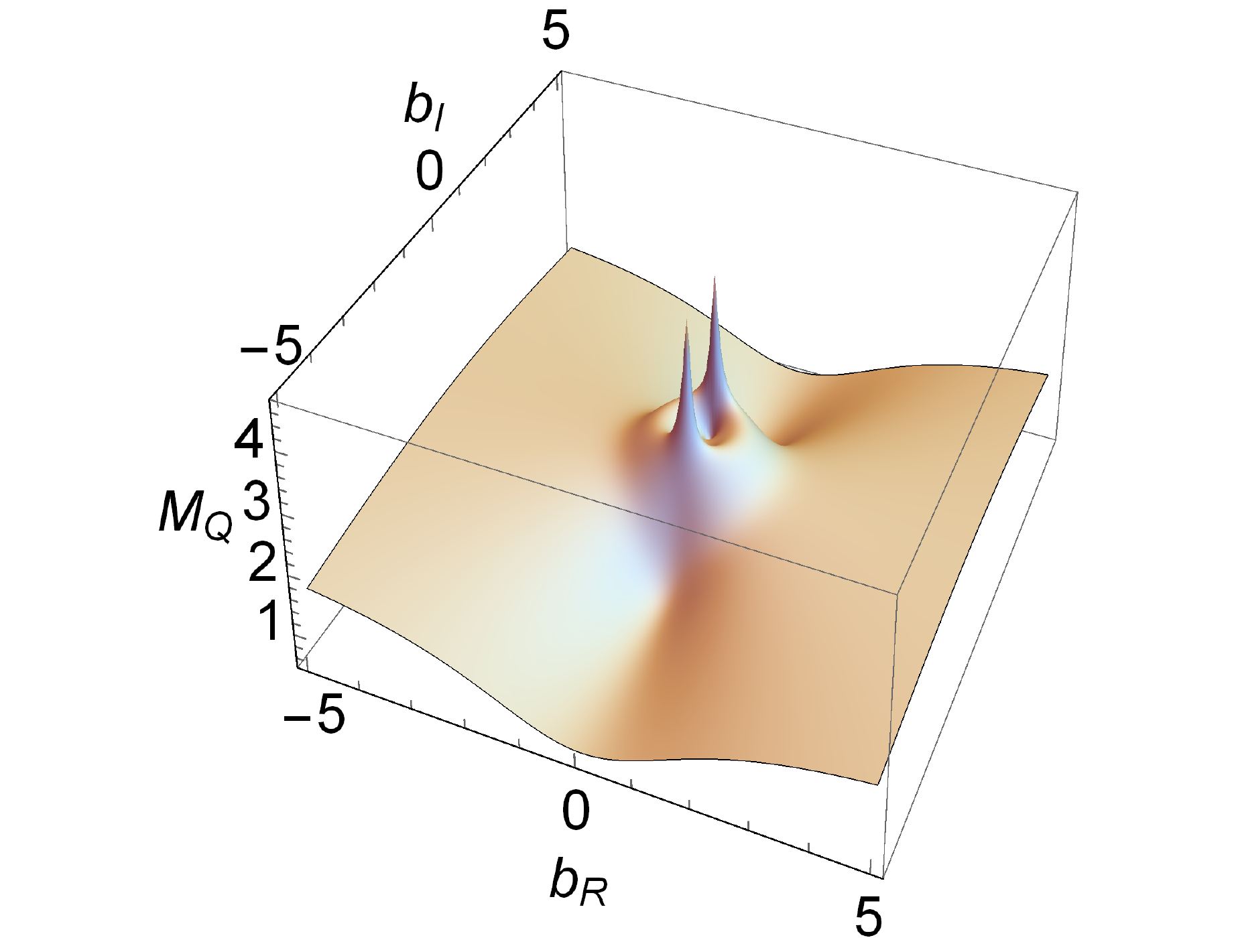}
(b)\includegraphics[scale=0.25]{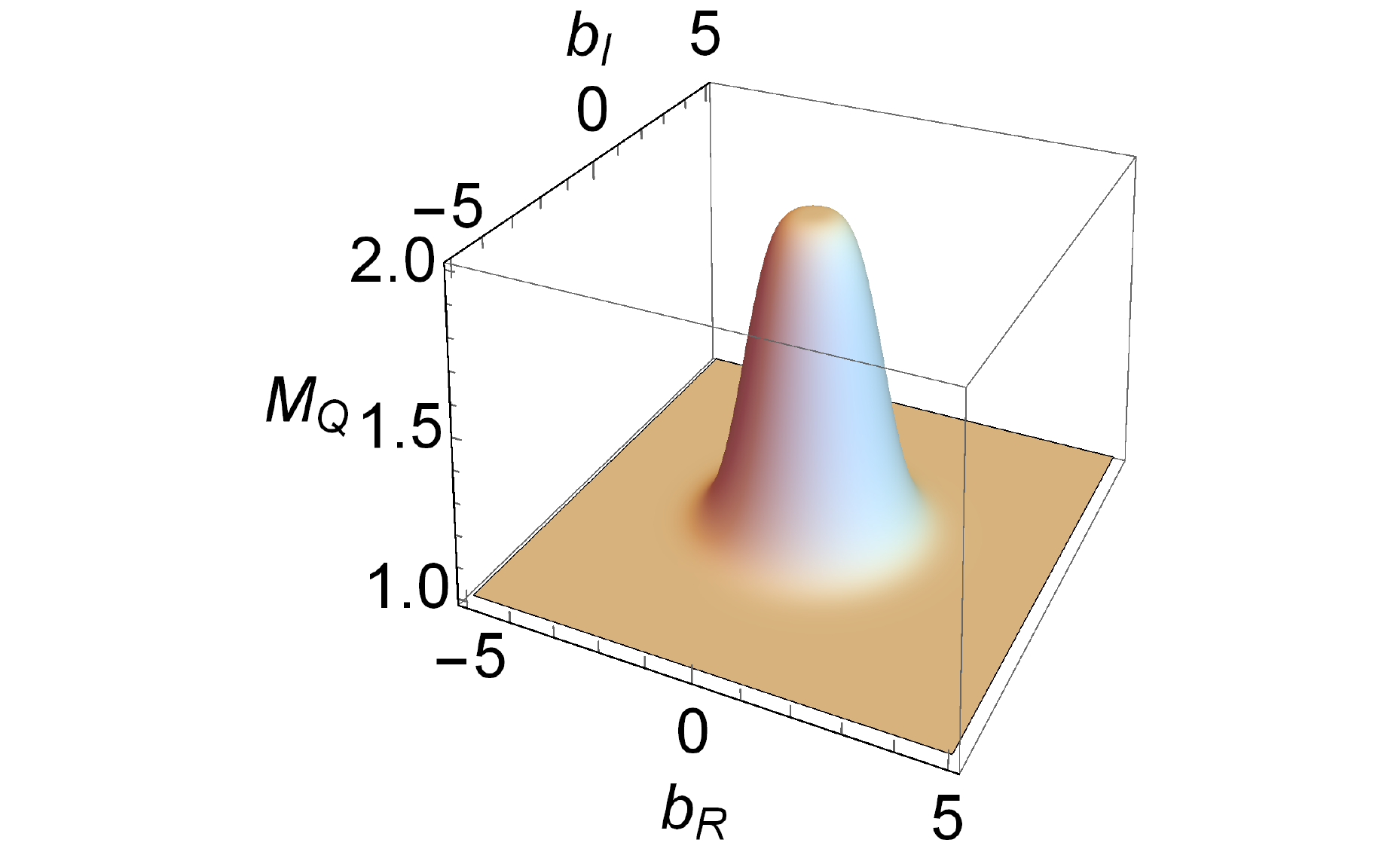}
(c)\includegraphics[scale=0.25]{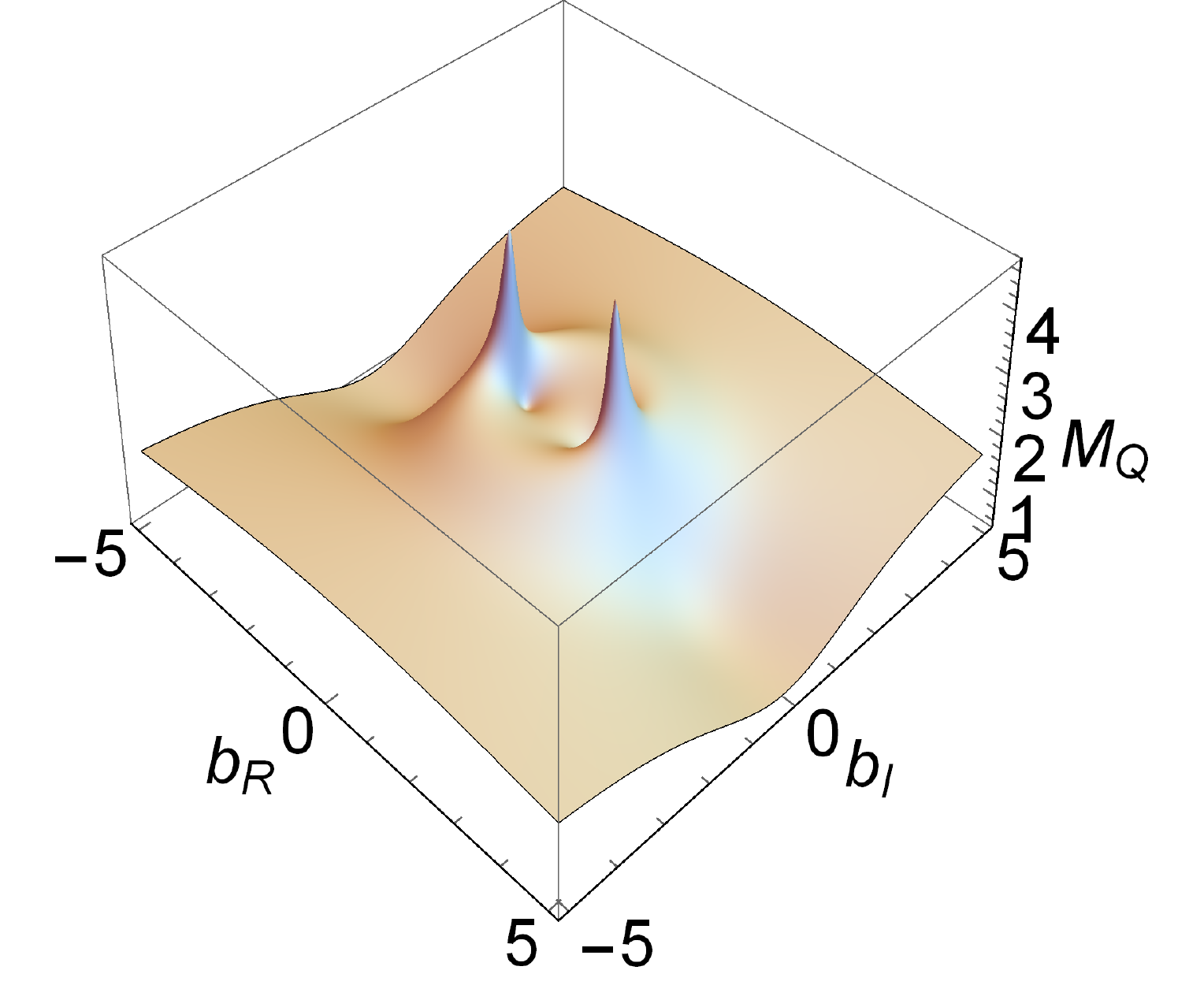}
(d)\includegraphics[scale=0.25]{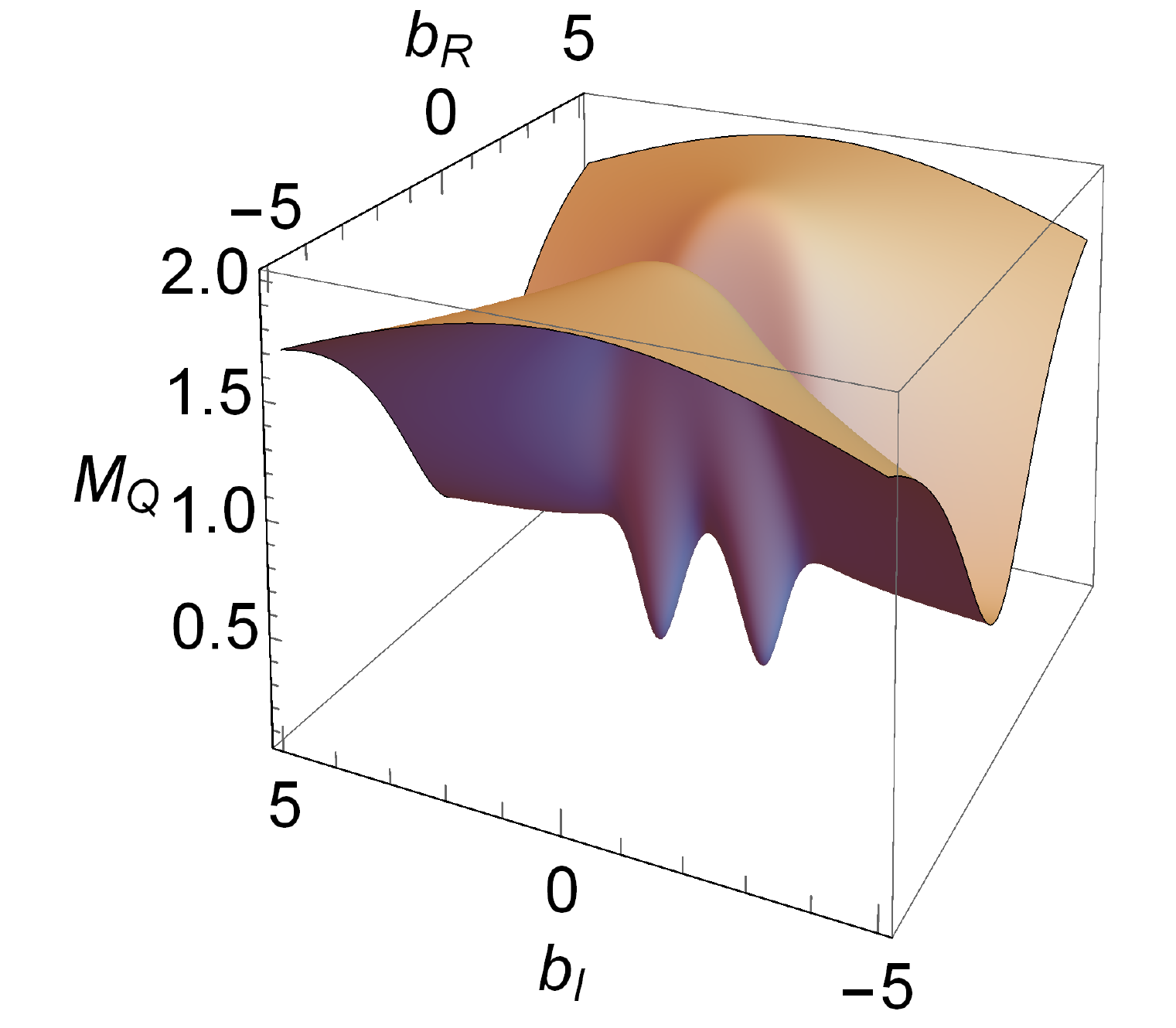}
(e)\includegraphics[scale=0.25]{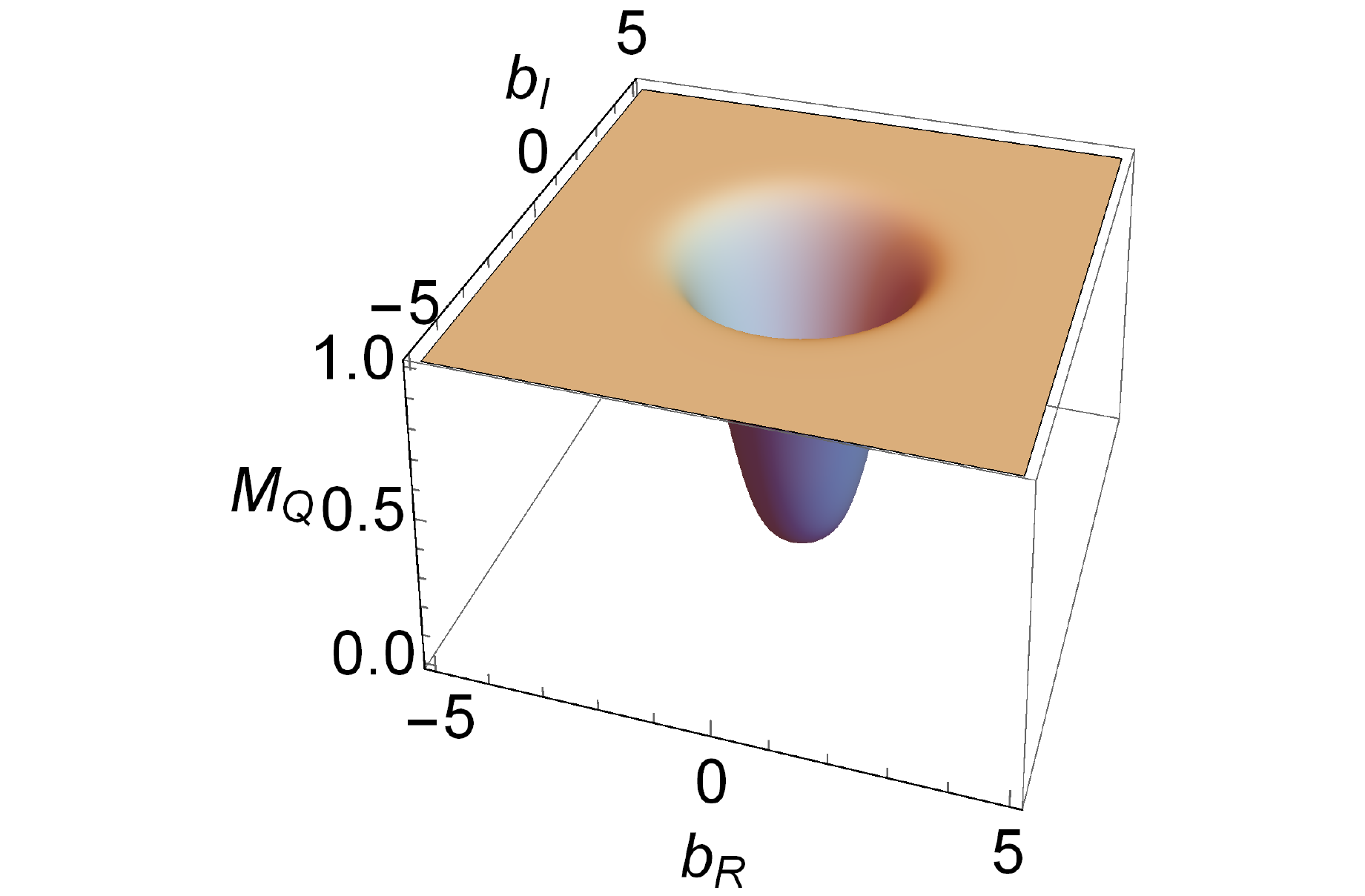}
(f)\includegraphics[scale=0.25]{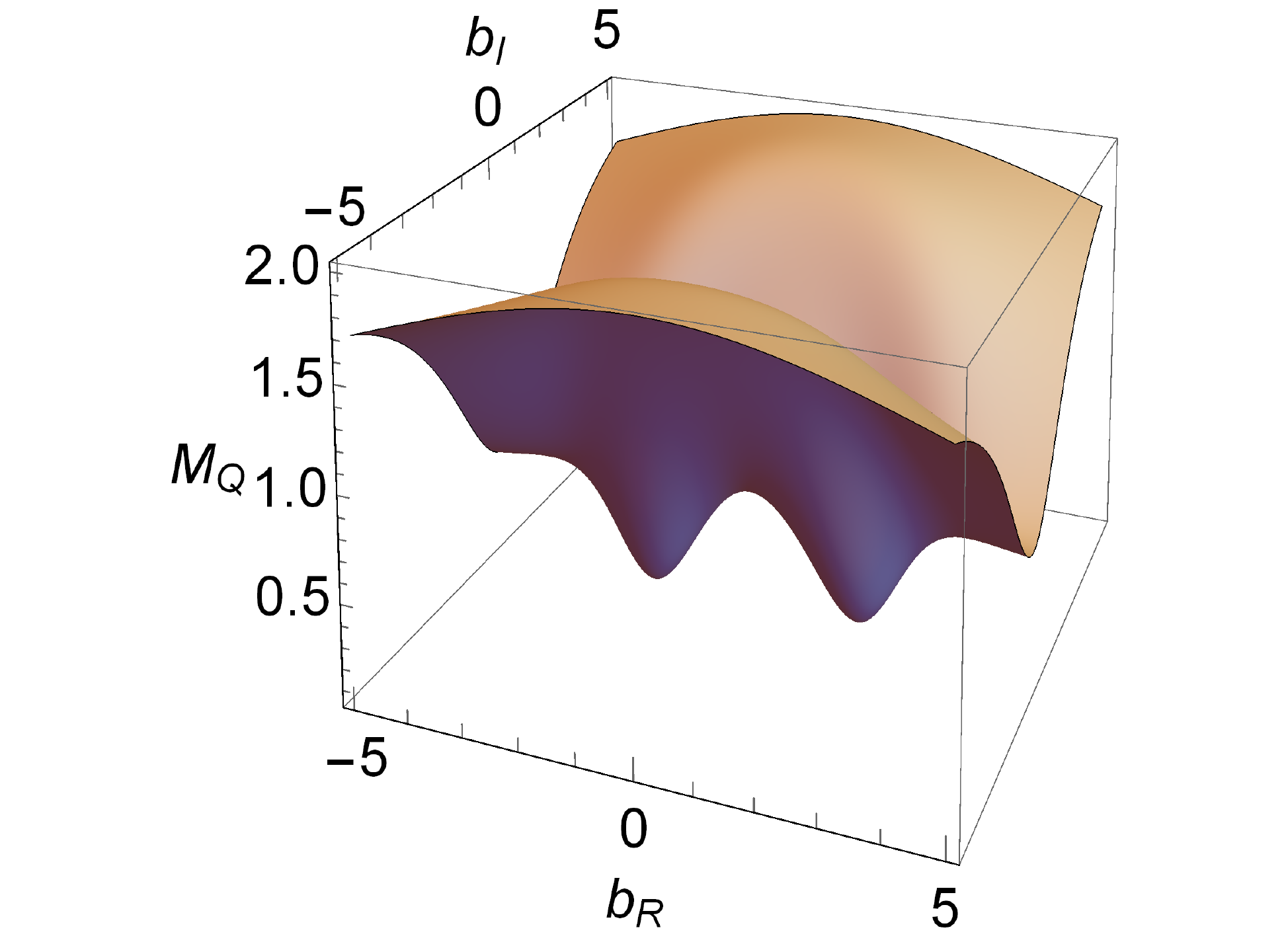}
\caption{Mandel parameter $M_Q$ as a function of the real and imaginary parts of the parameter $b=b_R+i b_I$, for the states associated to the cyclic group $C_2$, $\Psi_2^{(1)}(x)$ with (a) $a=1/4$, (b) $a=1/2$, (c) $a=1$; and for the state $\Psi_2^{(2)}(x)$ for (d) $a=1/4$, (e) $a=1/2$, and (f) $a=1$ . \label{mandel}}
\end{figure}
which have specific properties. In Fig.~\ref{mandel}, the Mandel parameter \cite{mandel} $M_Q=\langle (\Delta \hat{n})^2 \rangle/ \langle \hat{n} \rangle$ is shown for the cyclic Gaussian states of $C_2$ given in Eq.~(\ref{goga}). The figure was made taking into account three different $a$ parameters. 

A Mandel parameter $M_Q<1$ can be used to distinguish a subpoissonian from a superpoissonian  photon statistics ($M_Q>1$), or poissonian statistics $M_Q=1$. As it can be seen in the figure, the cyclic states can have subpoissonian distributions for certain regions of the parameter $b=b_R+i b_I$. As can be seen in the figure, the presence of this photon statistic is more prominent in the states associated to the second irreducible representation of the group $\Psi^{(2)}_2(x)$.

Similar to the states above, the ones associated to the cyclic group $C_3$ can be obtained. In Fig.~\ref{wignerc3}, the plots and contours for the Wigner function \cite{wigner}: $W_\psi(x,p)=\int \, dy \, \psi^*(x+y) \psi (x-y) e^{2ipy}/\pi$, can be seen. In the contour plots of the phase space ($p,x$) is noticed the symmetry of the state under the rotations with angles $0$, $2\pi/3$, and $4\pi/3$ with respect to the $x$ axis. It is also important to say that the Wigner functions depicted in the figure do not have inversion symmetry as they are only invariant under the rotations contained in the $C_3$ group.
%Figure 2
\begin{figure}
\centering
\includegraphics[scale=0.25]{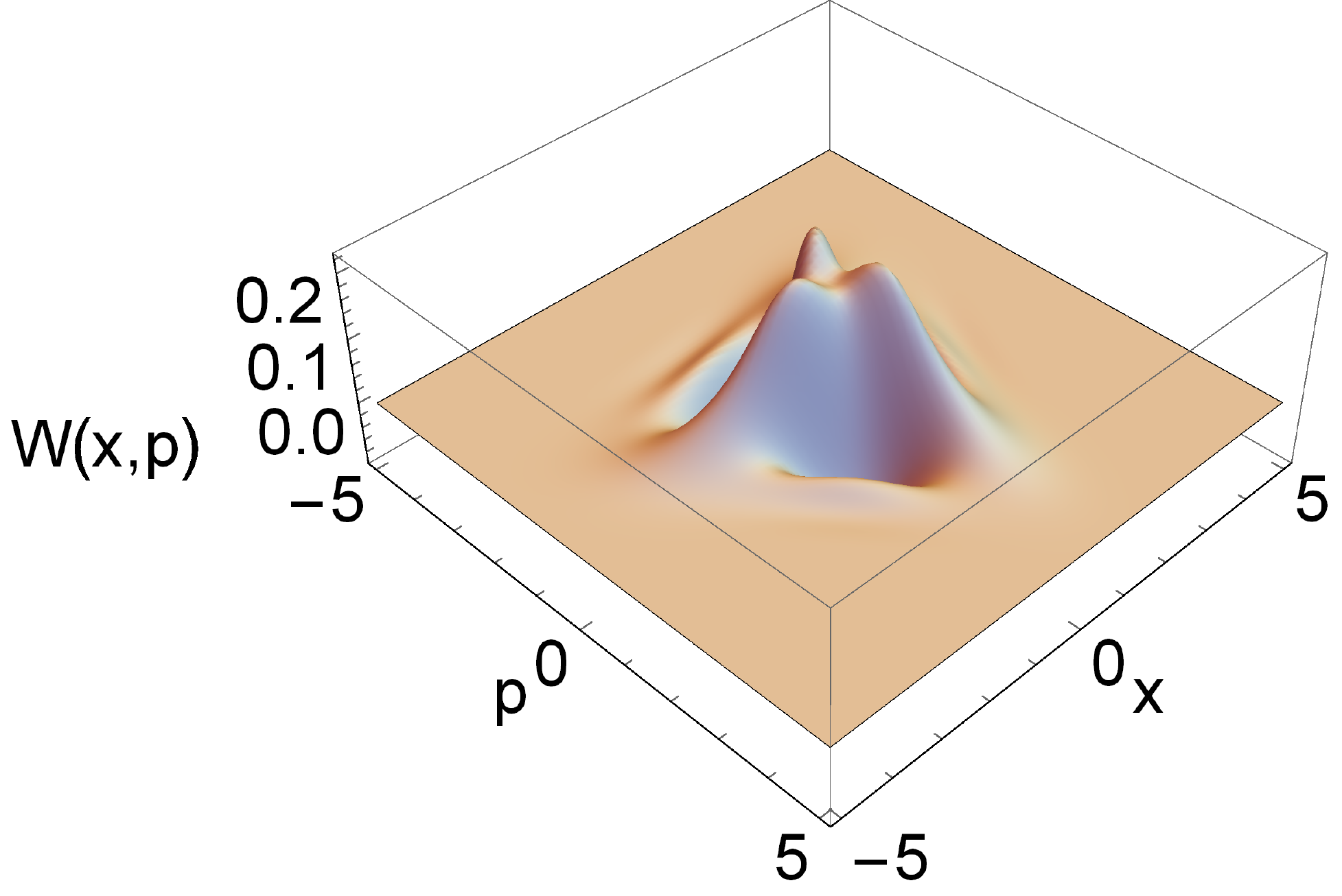}
\includegraphics[scale=0.25]{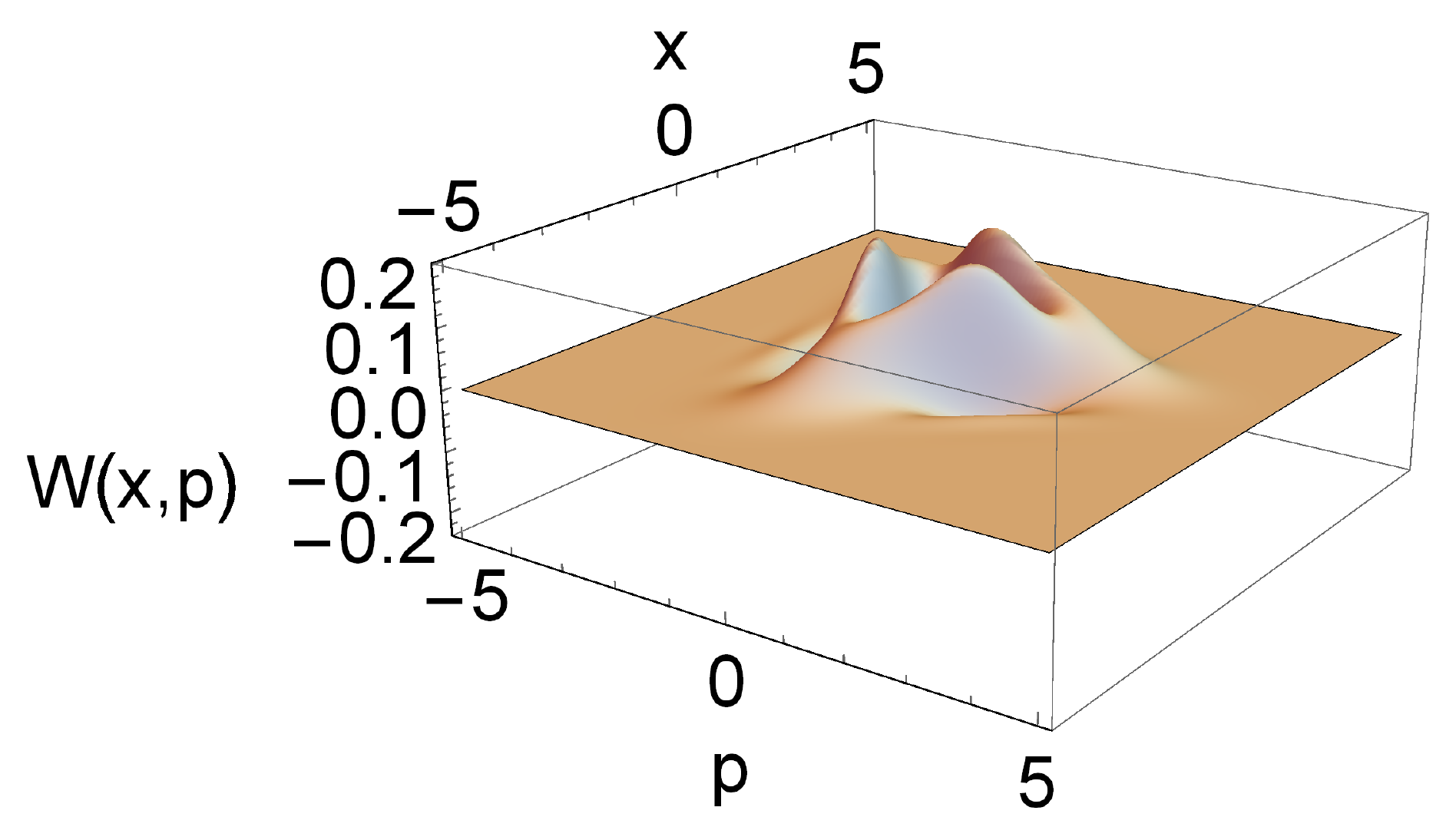}
\includegraphics[scale=0.25]{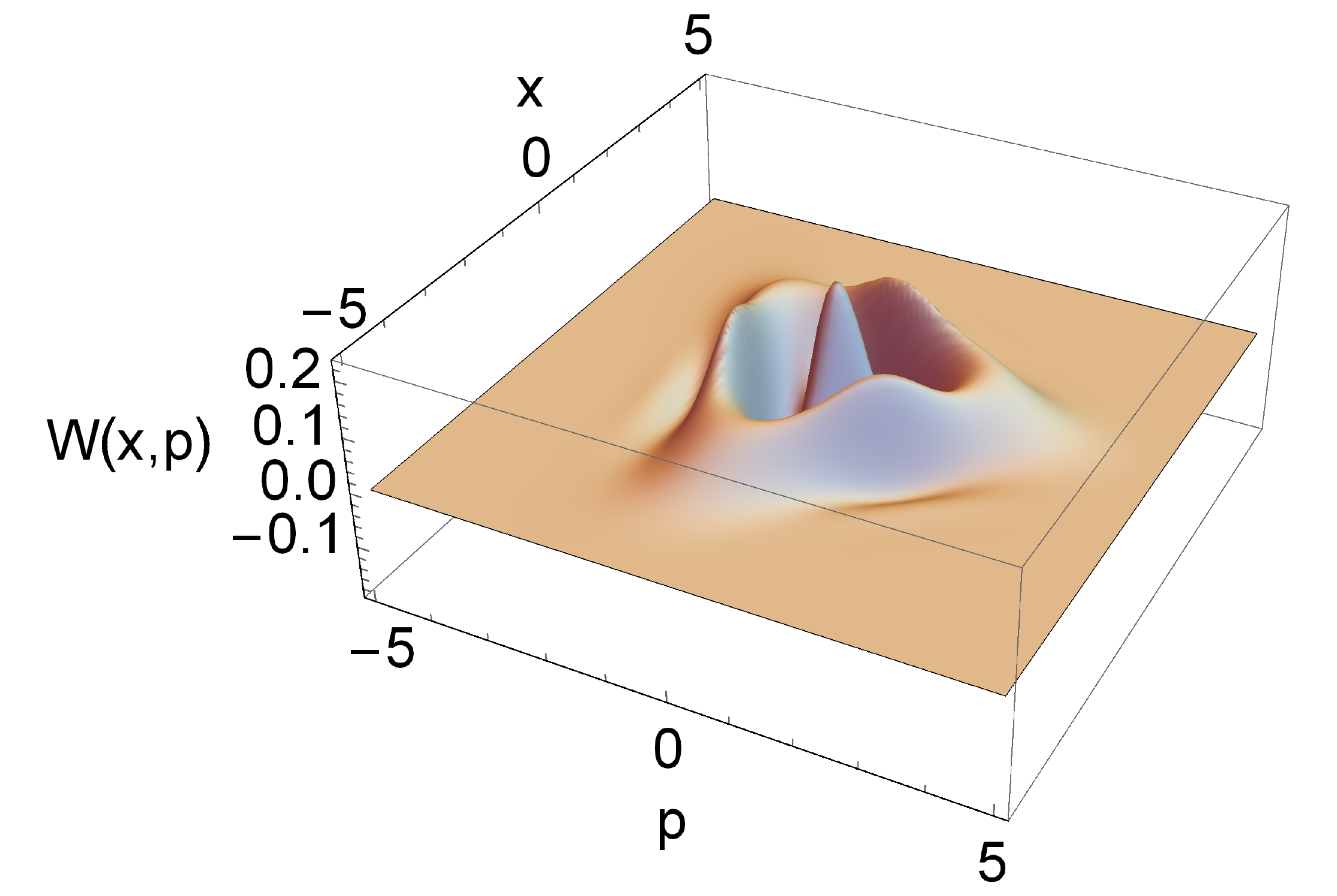}
\includegraphics[scale=0.22]{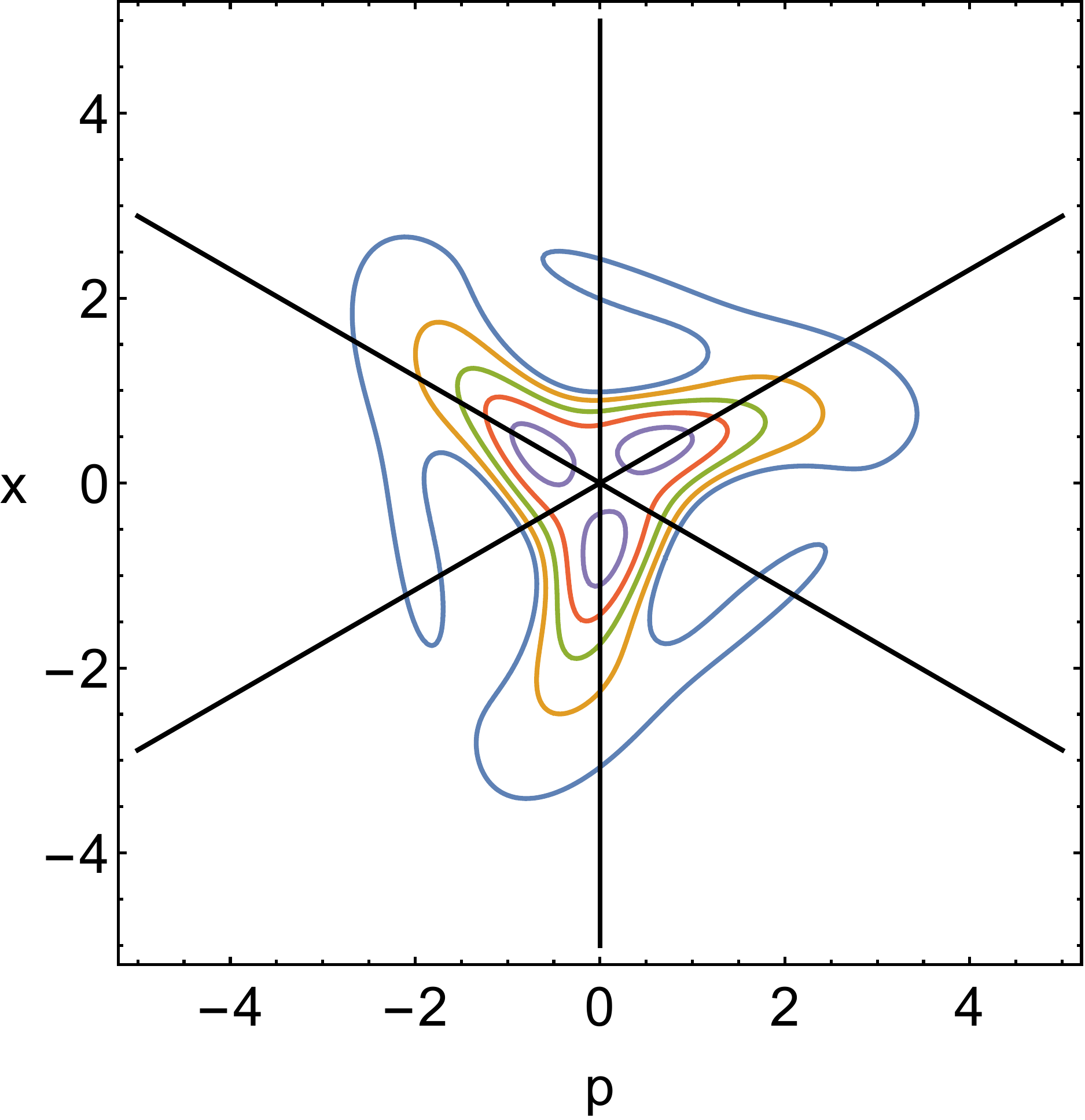}
\includegraphics[scale=0.22]{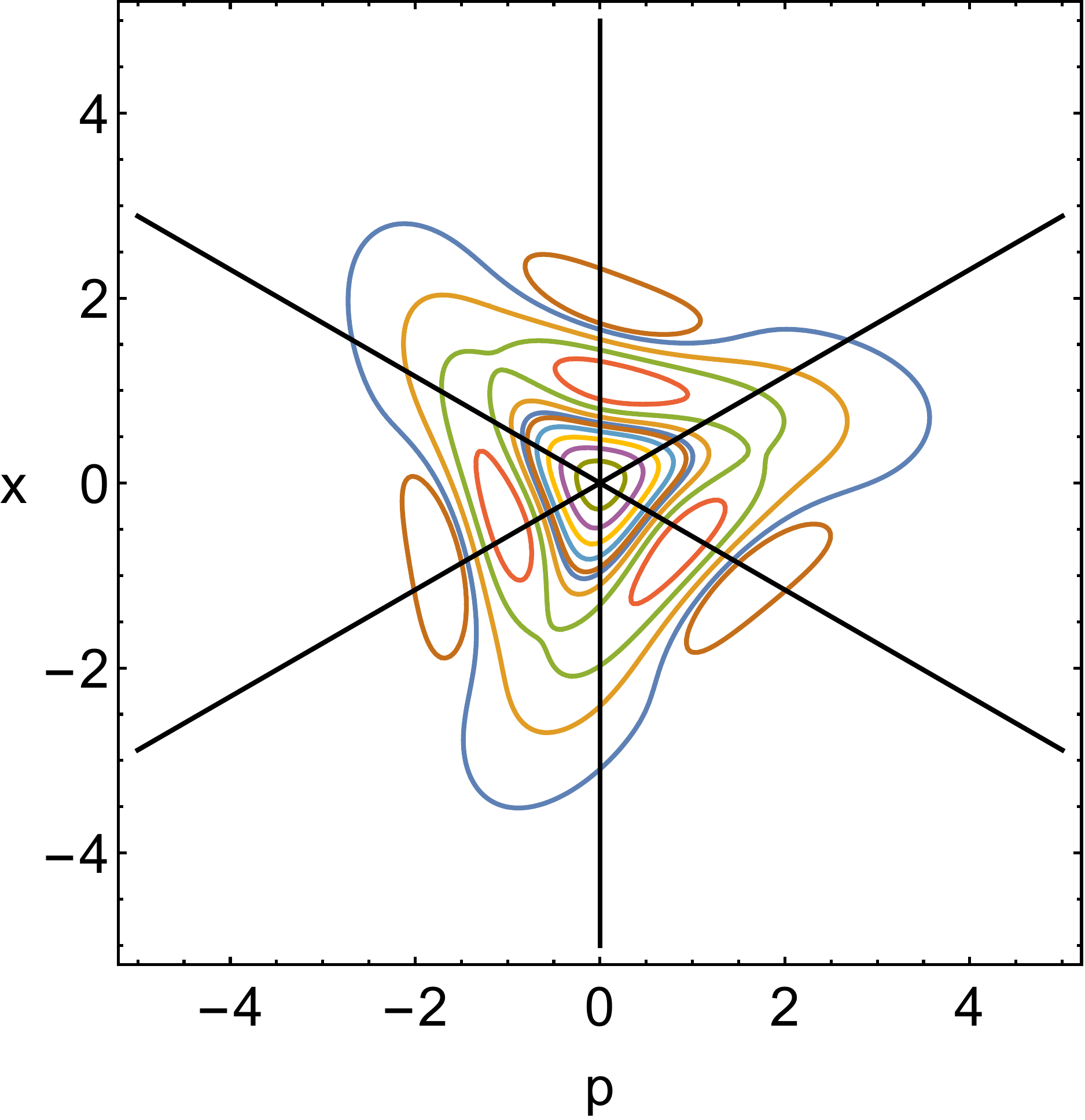}
\includegraphics[scale=0.22]{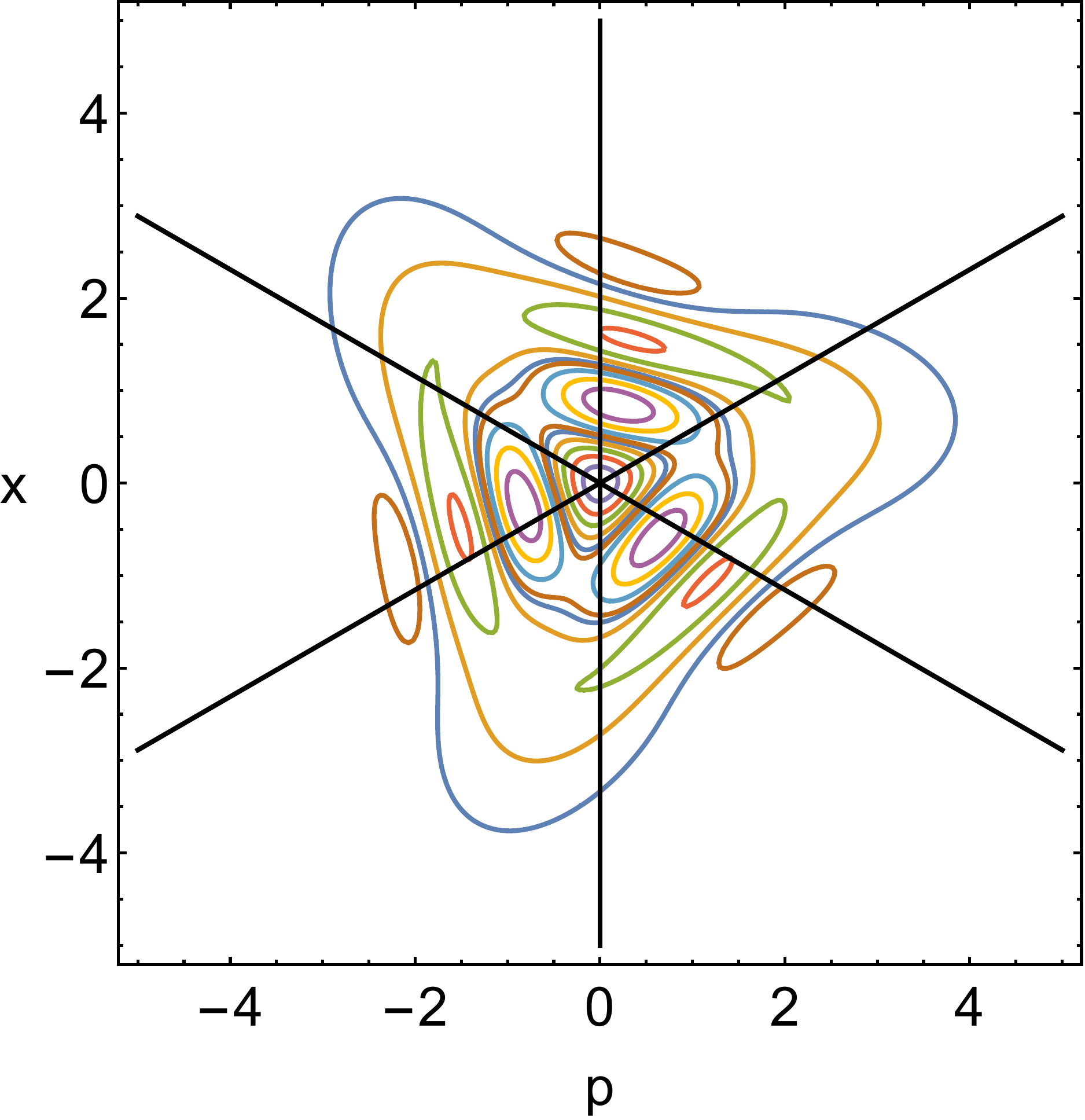}
\caption{Wigner functions and their contour plots for the cyclic Gaussian states associated to $C_3$ for the irreducible representations $\lambda=1$ (left), $\lambda=2$ (center), and $\lambda=3$ (right). For these figures the chosen parameters were $a=1$ and $b=\sqrt{2}(1+i)$. The black lines in the contour plots depict the symmetry axis associated to the $C_3$ group. \label{wignerc3}}
\end{figure}

\subsection{Circle symmetric states}
When one increases the degree of the cyclic group the obtained states described by our method must be invariant under more and more rotations in the phase space. It is known \cite{janszky} that there exist a correspondence between the circle symmetric states in the coherent case and the Fock number states. This lead us to the question, how do the generalized cyclic states associated to a very big number of symmetries look like? e.g., when the order of the cyclic group tends to infinite ($n\rightarrow \infty$), can they also be associated to the Fock states?. To answer these questions, one can notice that the Definition~\ref{defi1} of the cyclic states allow us to make a generalization in the case when the angle $\theta$, which determine the rotations $\hat{R}(\theta)$, becomes a continuous variable. In that case, the definition of the cyclic states becomes
\begin{equation}
\vert \psi_\infty^{(\lambda)}\rangle= \mathcal{N}_\lambda \int_0^{2\pi} d\theta \, e^{i\theta (\lambda-1)} e^{-i \theta \hat{n}} \vert \phi \rangle \, ,
\end{equation}
where we have an infinity number of irreducible representations, i.e., $\lambda \in \mathbb{Z}^+$. By means of the photon number decomposition of $\vert \phi \rangle=\sum_m A_m (\phi) \vert m \rangle$, one obtains
\[
\vert \psi_\infty^{(\lambda)}\rangle= \mathcal{N}_\lambda \sum_{m=0}^\infty \int_0^{2\pi} d\theta \, A_m (\phi) e^{i\theta (\lambda-1-m)} \vert m \rangle \, ,
\]
as the integral is equal to $2\pi$ times the Kronecker delta $\delta_{\lambda-1,m}$, we arrive to the result
\[
\vert \psi_\infty^{(\lambda)}\rangle= 2\pi \mathcal{N}_\lambda A_{\lambda-1}(\phi) \vert \lambda-1 \rangle \, ,
\]
which, after the renormalization process, we notice corresponds to the number state 
\begin{equation}
\vert \psi_\infty^{(\lambda)}\rangle=\vert \lambda-1 \rangle \, .
\end{equation}
We point out that this expression for the circle cyclic states is consistent with the erasure map of the state $\vert \phi \rangle$, as in principle we need to erase all the different states but the one that satisfies the condition $m-\lambda+1=0$. This result lead us to the conclusion that the cyclic superposition ($n \rightarrow \infty$) of any state which is noninvariant under any rotation in the phase space, is equal to a Fock state. This, regardless of the initial, noninvariant state $\vert \phi \rangle$ that we take into consideration. We would like to emphasize that in order of this property to be true, the state under consideration $\vert \phi \rangle$ must be noninvariant under all possible rotations in the phase space. This implies that $\vert \phi \rangle$ must be expressed by an infinite sum of the photon number states $\vert m \rangle$ with a nonzero probability amplitude $A_m(\phi)$. To show this we can take as an example the $C_2$ group. In order for a state $\vert \phi \rangle$ to be noninvariant under the $C_2$ rotation, it should be made by the superposition of at least two states $\vert m \rangle$ and $\vert n \rangle$, $m$ being even and $n$ being odd ($m,n\in \mathbb{Z}^+$). In the case of $C_3$ we need at least three states $\vert m \rangle$, $\vert n \rangle$, and $\vert l \rangle$ such mod$(m,3)=0$, mod$(n,3)=1$, and mod$(l,3)=2$ ($m,n,l\in \mathbb{Z}^+$). By the extension of this argument, we must need an infinite number of photon states in order for $\vert \phi \rangle$ to be an noninvariant state under $C_\infty$. As examples of this type of states one can name the coherent, the squeezed coherent, the non-centered Gaussian, and any noninvariant, continuous variable state. 

To show that the superposition of several rotations of an initial continuous variable system can form a Fock state one can take as an example the Gaussian state of Eq. (\ref{gaus}) with $a=1$, $b=\sqrt{6}+2 i$. In Fig.~\ref{wigcirc} are shown the Wigner functions and their contours for the cyclic states associated to the first irreducible representation of $C_n$ for $n=10$ (left), $n=15$ (center), and $n=20$ (right). Here, one can see how the cyclic states for a long degree order are more and more alike to the vacuum state $\vert 0 \rangle$. Additionally to this, it can be checked that for a given irreducible representation of the cyclic group, a different photon state can be formed for a sufficient large number $n$, i.e., the cyclic group degree.
%Figure 3
\begin{figure}
\centering
\includegraphics[scale=0.30]{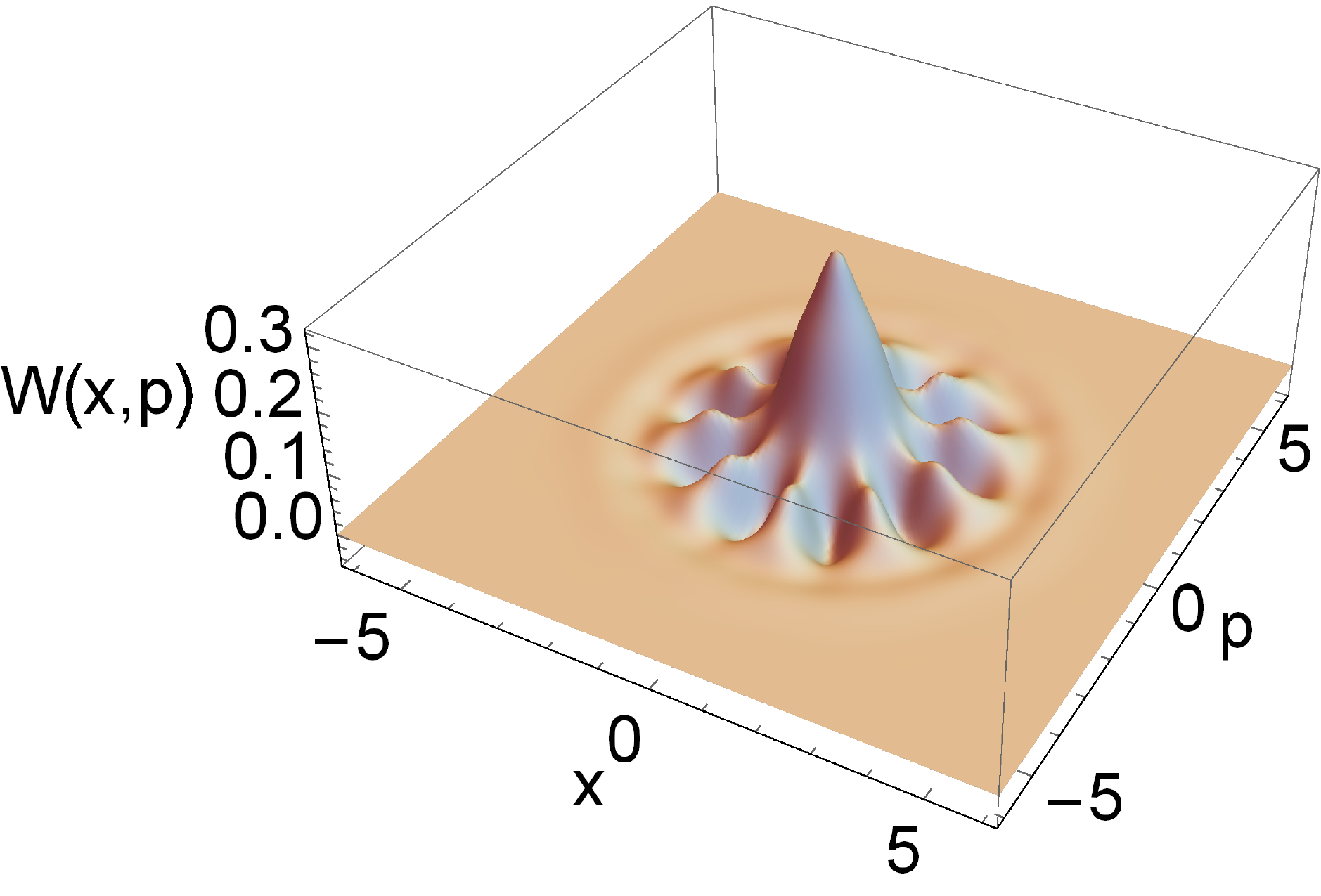}
\includegraphics[scale=0.30]{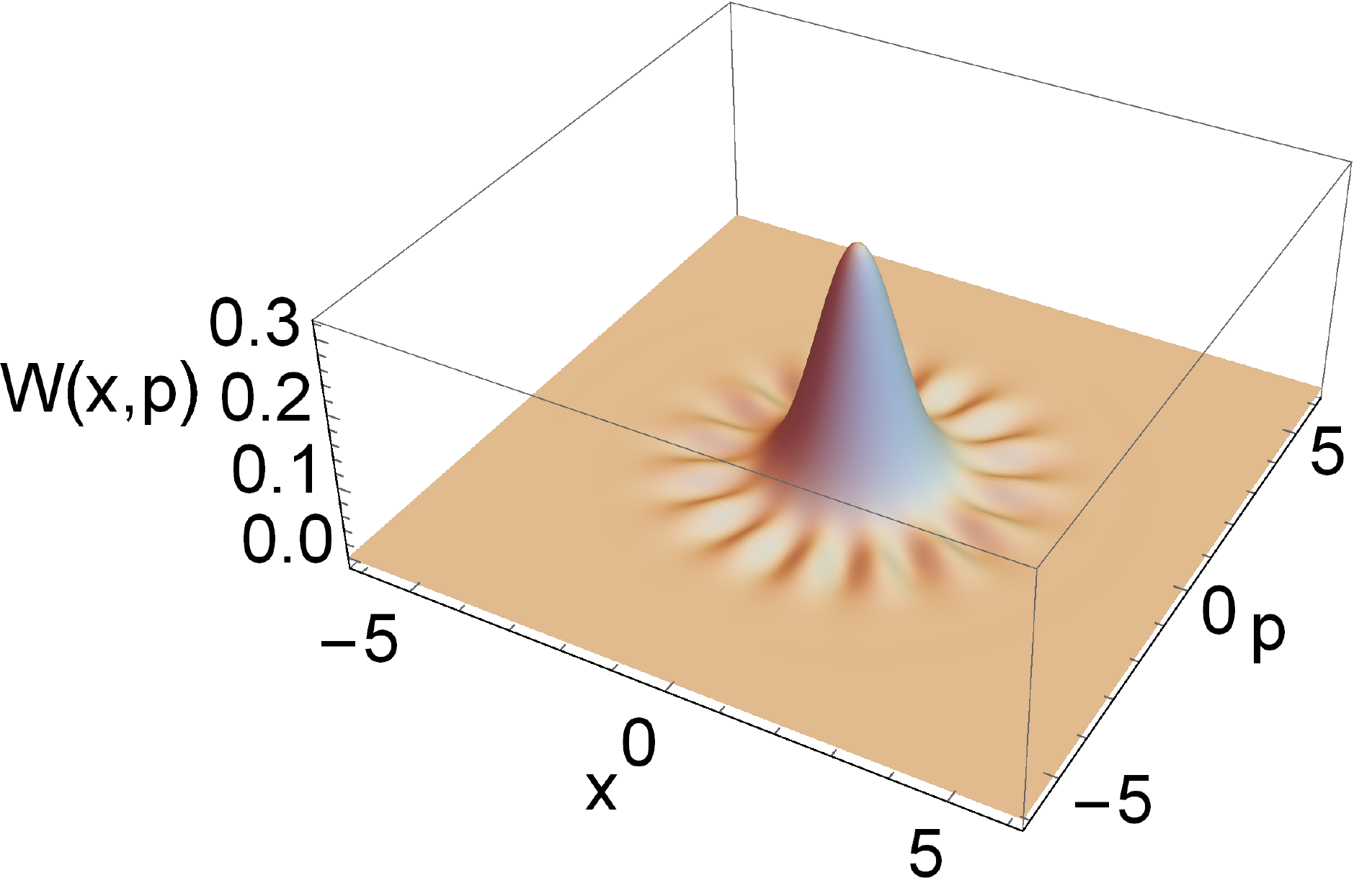}
\includegraphics[scale=0.30]{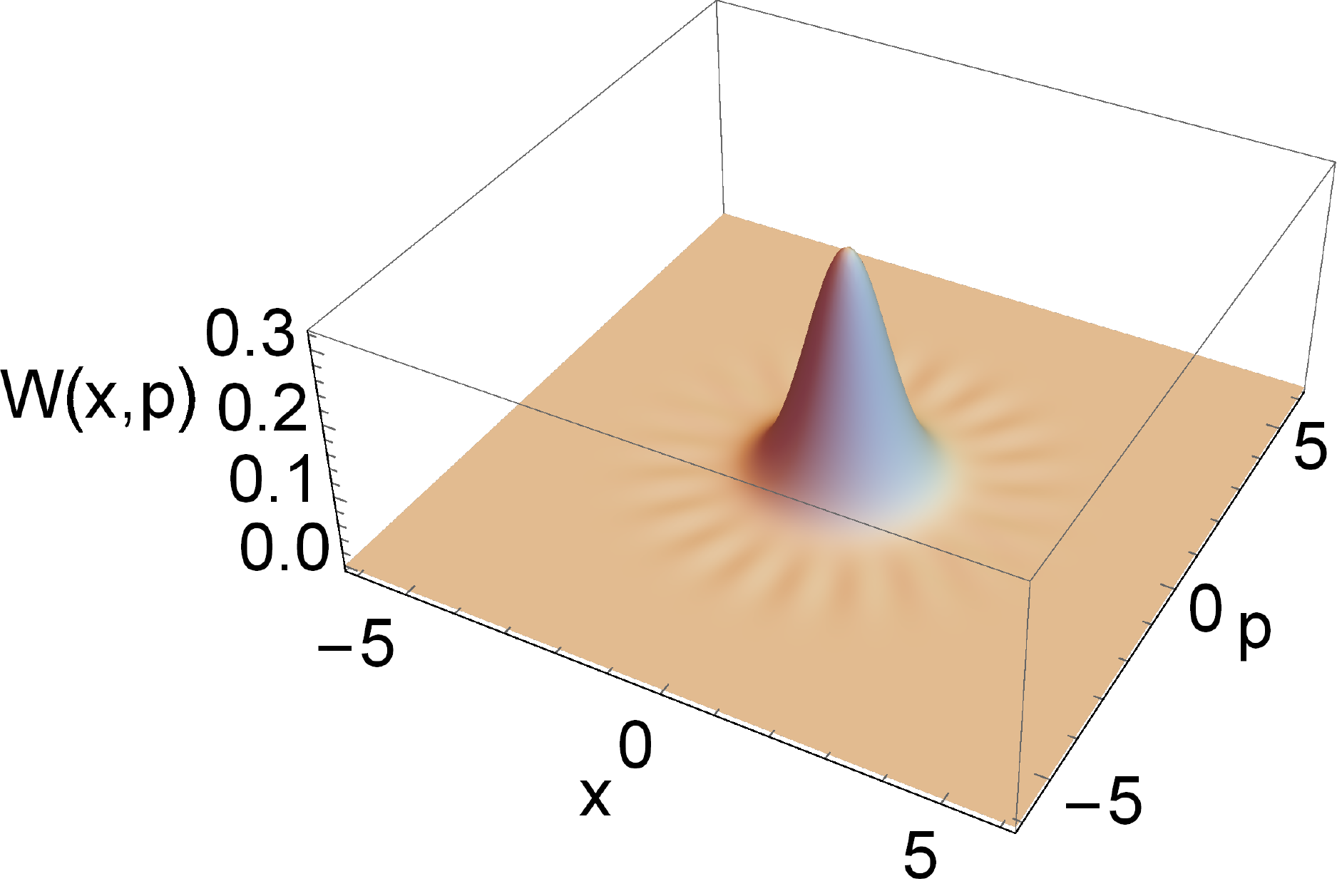}\\
\includegraphics[scale=0.35]{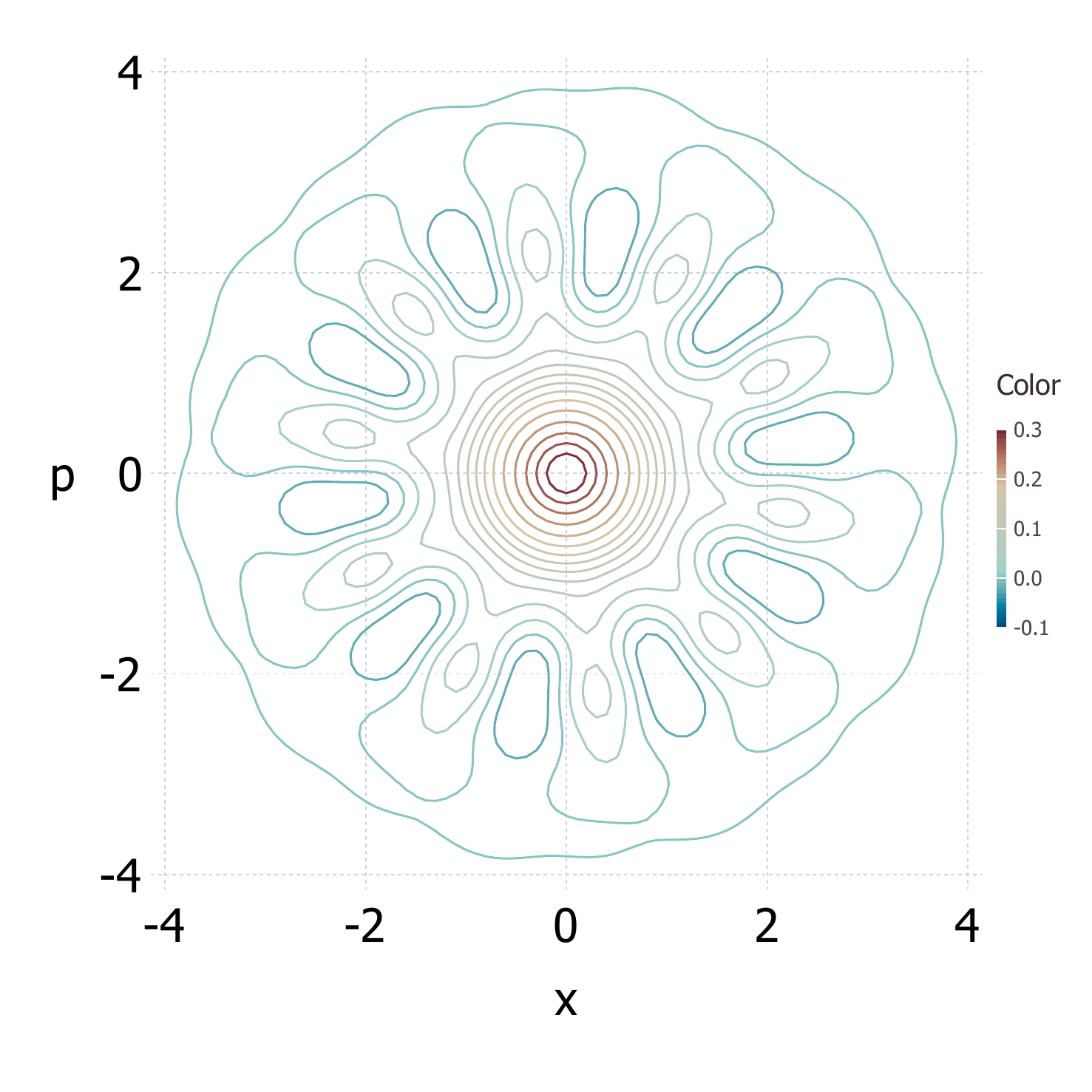}
\includegraphics[scale=0.35]{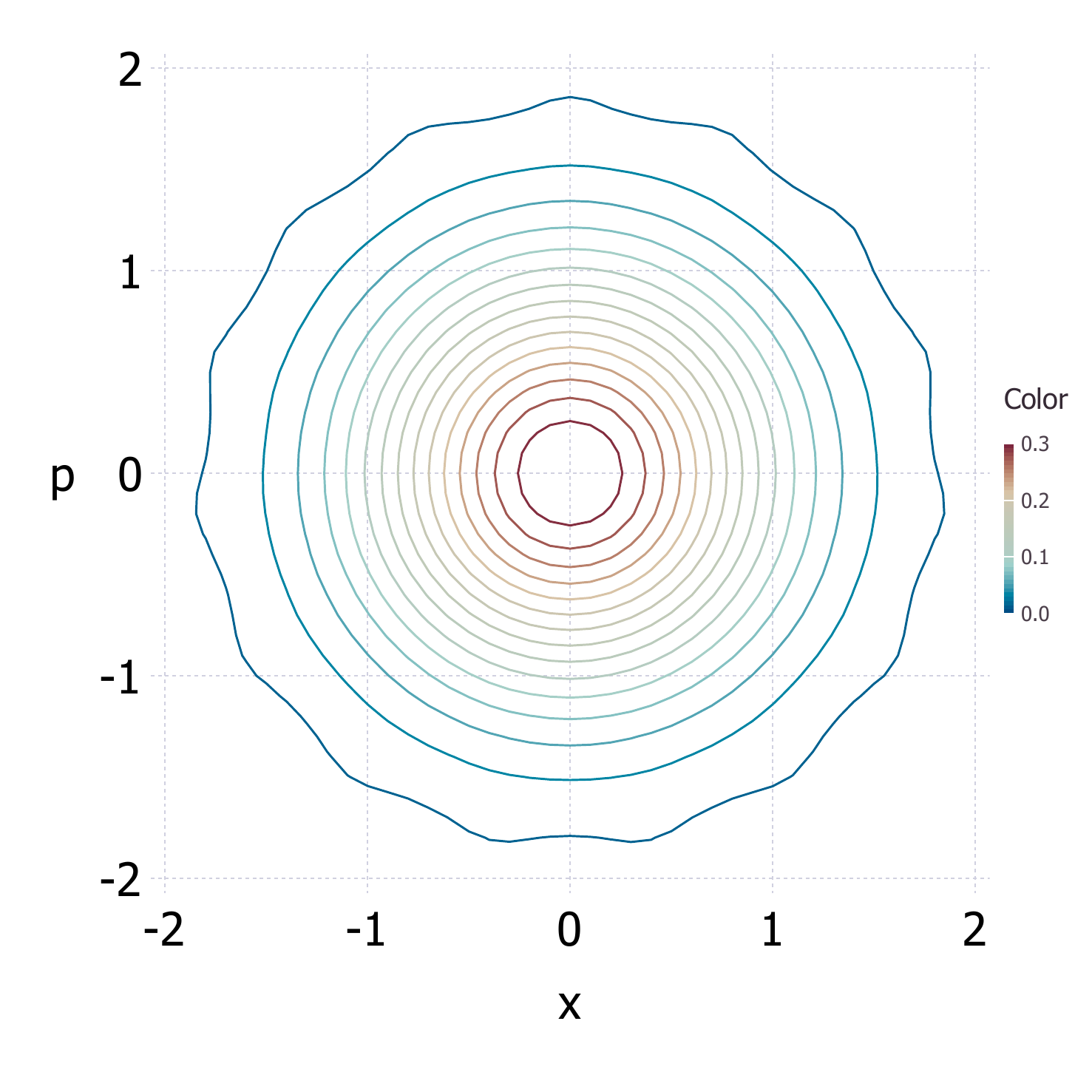}
\includegraphics[scale=0.35]{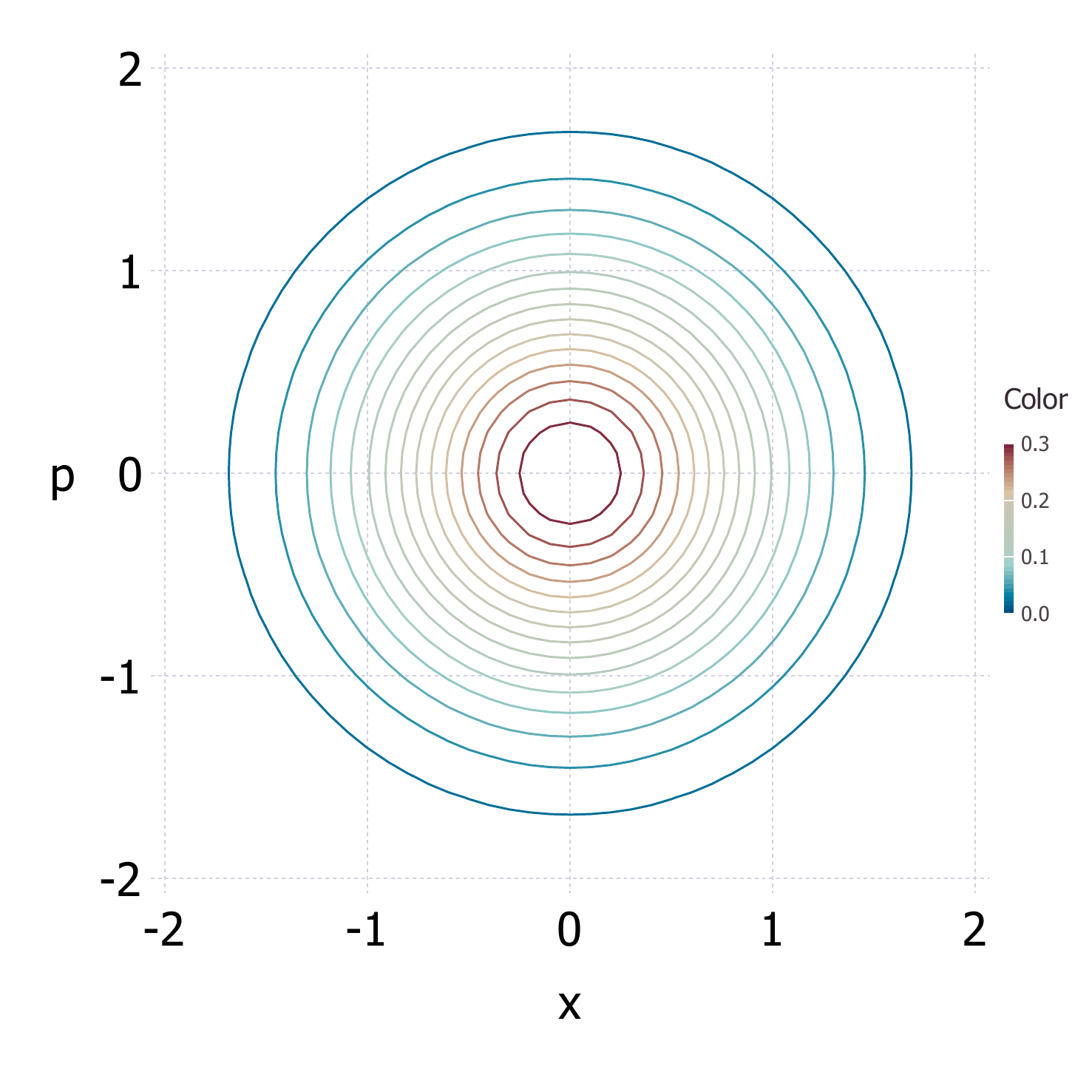}

\caption{Wigner functions and their contour plots for the cyclic Gaussian state of $C_n$ for the irreducible representation $\lambda=1$ for (a) $n=10$ (right), (b) $n=15$ (center), and (c) $n=20$ (left). In all the plots we took the parameters $a=1$ and $b=\sqrt{6}+2 i$. \label{wigcirc}}
\end{figure}

\section{Cyclic group density matrices.}

The previous discussion about the properties of the erasure map and its relation with the states associated to the cyclic groups can be extended to any kind of state which is not invariant under the rotation operation. For example, one can can think in a density matrix which may correspond to a mixed state $\hat{\rho}$ and define the following cyclic density matrices

\begin{mydef}
Let $\hat{\rho}$ be a density matrix with at least one of its mean quadrature components ($\hat{x}=(\hat{a}+\hat{a}^\dagger)/\sqrt{2}$, $\hat{p}=i(\hat{a}^\dagger-\hat{a})/\sqrt{2}$) different from zero, i.e., ${\rm Tr}(\hat{\rho}\, \hat{x})\neq 0$, or ${\rm Tr}(\hat{\rho}\, \hat{p})\neq 0$. Then the state associated to the irreducible representation $\lambda$ of  the cyclic group $C_n$ is defined as
\begin{equation}
\hat{\rho}^{(\lambda)}_n=\mathcal{N}_\lambda \sum_{r,s=1}^n \chi_n^{(\lambda)} (g_r) \chi_n^{*(\lambda)}(g_s) \hat{R}(\theta_r) \hat{\rho} \hat{R}^\dagger(\theta_s) \, ,
\label{rhoo}
\end{equation}
where $\chi_n^{(\lambda)} (g_r)$ is the character for the  group element $g_r$, $\hat{R}(\theta_r)=\exp{(-i \theta_r \hat{n})}$, and
\[
\mathcal{N}_\lambda^{-1}=\sum_{r,s=1}^n \chi_n^{(\lambda)} (g_r) \chi_n^{*(\lambda)}(g_s) \, {\rm Tr}(\hat{R}(\theta_r) \hat{\rho} \hat{R}^\dagger(\theta_s)) \, .
\]
\end{mydef}
These type of density matrices have the same properties of the cyclic states as being invariant up to a phase under the rotations in the cyclic group. Also, they have a photon distribution were not all the photon numbers are present as they can be obtained by the elimination of certain Fock states. To show this, one can follow an analogous procedure as in Theorem (\ref{teo2}). Let us suppose $\hat{\rho}=\sum_{m,m'=0}^\infty A_{m,m'}(\hat{\rho}) \vert m \rangle \langle m' \vert$, with ${\rm Tr}(\hat{\rho})=\sum_{m=0}^\infty A_{m,m}(\hat{\rho})=1$. This expression together with Eqs. (\ref{chi}) and (\ref{rhoo}) allow us to rewrite $\hat{\rho}_n^{(\lambda)}$ as follows
\[
\hat{\rho}_n^{(\lambda)}=\mathcal{N}_\lambda \sum_{m,m'=0}^\infty A_{m,m'}(\hat{\rho})\sum_{r,s=1}^n \mu_n^{(\lambda-1)(r-1)} \mu_n^{(\lambda-1)(1-s)} e^{-i \theta_r m} e^{i \theta_s m'} \vert m \rangle \langle m' \vert \, ,
\]
by the use of the definition of $\theta_j=2\pi (j-1)/n$ and Theorem \ref{tt1}, we can perform the sums over the $r$ and $s$ parameters. Those sums are
\begin{eqnarray}
\sum_{r=1}^n \mu_n^{(\lambda-1-m)r}&=&n \, \delta_{{\rm mod}(\lambda-1-m,n),0} \, , \nonumber \\
\sum_{s=1}^n \mu_n^{-(\lambda-1-m')s} &=& n \, \delta_{{\rm mod}(\lambda-1-m',n),0} \, , 
\end{eqnarray}
then we finally can write the cyclic state density matrices as follows
\[
\hat{\rho}_n^{(\lambda)}=\mathcal{N}_\lambda \, n^2 \sum_{m.m'=0}^\infty A_{m,m'}(\hat{\rho}) \, \mu_n^{m'-m}\,  \delta_{{\rm mod}(\lambda-1-m,n),0} \, \delta_{{\rm mod}(\lambda-1-m',n),0} \, \vert m \rangle \langle m' \vert \, ,
\]
as the delta functions imply that $\lambda-1-m$ and $\lambda-1-m'$ should be a multiple of $n$, then $\lambda-1-m=\eta n$, and $\lambda-1-m'=\xi n$ and then $m'-m=-(\xi+\eta)n$ is also a multiple of $n$. From these properties, we can conclude that $\mu_n^{m'-m}=1$ and finally arrive to the expression for the cyclic density matrix
\begin{equation}
\hat{\rho}_n^{(\lambda)}=\mathcal{N}_\lambda \, n^2 \sum_{m.m'=0}^\infty A_{m,m'}(\hat{\rho}) \, \delta_{{\rm mod}(\lambda-1-m,n),0} \, \delta_{{\rm mod}(\lambda-1-m',n),0} \, \vert m \rangle \langle m' \vert \, ,
\label{rhot}
\end{equation}
this property is summarized in the following theorem:
\begin{theorem}
Let  $n$ and $\lambda$ be two positive integers with $\lambda\leq n$, and $\hat{\rho}_{n,\lambda} $ be the renormalized state obtained after the elimination of the number states operators $\vert m \rangle \langle m' \vert$ in $\hat{\rho}=\sum_{m,m'=0}^\infty A_{m,m'} (\hat{\rho}) \vert m \rangle \langle m' \vert$, which do not satisfy the conditions ${\rm mod}(\lambda-1-m,n)=0$ and ${\rm mod}(\lambda-1-m',n)=0$, then $\hat{\rho}_{n,\lambda} (\phi)\rangle$ is equal to the cyclic state $\hat{\rho}_n^{(\lambda)}$.
\label{teo3}
\end{theorem}

It is noteworthy to see that from Eq. (\ref{rhot}) and the property $m'-m$ being a multiple of $n$, we can immediately show that the cyclic density matrices are invariants over the rotations in the cyclic groups. In other words, the density matrix $\hat{\rho}_n^{(\lambda)}$ after the rotation $\hat{R}(\theta_j)$, i.e.,
\[
\hat{R}(\theta_j)\hat{\rho}_n^{(\lambda)} \hat{R}^\dagger (\theta_j)=\mathcal{N}_\lambda \, n^2 \sum_{m.m'=0}^\infty A_{m,m'}(\hat{\rho}) \, \delta_{{\rm mod}(\lambda-1-m,n),0} \, \delta_{{\rm mod}(\lambda-1-m',n),0} \, \mu_n^{(m'-m)(j-1)}\, \vert m \rangle \langle m' \vert \, ,
\]
is equal to the initial density matrix, so finally one can establish
\[
\hat{R}(\theta_j)\hat{\rho}_n^{(\lambda)} \hat{R}^\dagger (\theta_j)=\hat{\rho}_n^{(\lambda)} \, .
\]

As in the case of the pure cyclic states, the photon number distribution of the cyclic density matrices contains only some of the numbers states. Given that the different states associated to the cyclic group $C_n$ are made with different photon number states, we can conclude that the cyclic density matrices form an orthogonal set.

\section{Example: Calculation of the entanglement in a bipartite state.}

As an example of the applications of the cyclic states we show that this type of states can be used to describe a continuous variable system in a discrete way, and that this discrete form can lead to an easier calculation of parameters, such as the entanglement between parts in a bipartite system. Suppose a two mode state made entirely of the group of rotation states $\{\vert \phi_r \rangle_1, \vert \varphi_r \rangle_2 ; r=1, \ldots, n\}$ for modes 1 and 2 respectively, e.g. the state
\begin{equation}
\vert T\rangle= \sum_{r=1}^n c_r \vert \phi_r \rangle_1 \vert \varphi_r \rangle_2 \, , \quad \sum_{r,r'=1}^n c_r c_{r'}^* \langle \phi_{r'} , \varphi_{r'} \vert \phi_r, \varphi_r\rangle =1 \, .
\label{tstate}
\end{equation}

As the states $\vert \phi_r \rangle=\hat{R}(\theta_r)\vert \phi \rangle$, $\vert \varphi_r \rangle=\hat{R}(\theta_r)\vert \varphi \rangle$ can be general then they might not be orthogonal. On the other hand, the cyclic states generated by these states form an orthogonal set. Most importantly, as there exist the same number of cyclic states $\vert \psi_n^{(\lambda)} (\phi)\rangle$ and $\vert \psi_n^{(\lambda)}(\varphi) \rangle$ as the number of rotated states  $\vert \phi_r \rangle$ and $\vert \varphi_r \rangle$, then one can obtain the rotated states in terms of the cyclic, orthogonal ones. To obtain these expressions one must obtain the inverse relation of Eq. (\ref{ccy})
\[
\vert \psi_n^{(\lambda)}(\phi) \rangle=\mathcal{N}_\lambda \sum_{r=1}^n \mu_n^{(\lambda-1)(r-1)} \vert \phi_r \rangle \, ,
\]
to do that, one can treat the characters of the group as a matrix $M_{jk}=\mu_n^{(j-1)(k-1)}$, which has an inverse matrix $M_{jk}^{-1}=\mu_n^{(1-j)(k-1)}/n$. By this expression one can obtain the inverse equation
\begin{equation}
\vert \phi_r \rangle = \frac{1}{n \, \mathcal{N}_\lambda} \sum_{\lambda=1}^n \mu_n^{(1-r)(\lambda-1)} \vert \psi_n^{(\lambda)}(\phi) \rangle \, .
\label{innv}
\end{equation}
By substituting this expression and an analogous expression for $\vert \varphi_r \rangle$ into the two-mode state $\vert T \rangle$, one obtains
\[
\vert T \rangle =\frac{1}{n^2}\sum_{r=1}^n c_r \sum_{\lambda,\lambda'=1}^n \frac{1}{\mathcal{N}_\lambda \mathcal{N}_{\lambda'}}\mu_n^{(1-r)(\lambda-1)} \mu_n^{(1-r)(\lambda'-1)} \vert \psi_n^{(\lambda)}(\phi) \rangle_1 \vert \psi_n^{(\lambda')}(\varphi) \rangle_2 \, .
\]
From this expression is possible to calculate the partial density matrices for each mode in the bipartite state. For this we obtain the total density matrix and perform the partial trace operation. Finally, arriving to
\begin{eqnarray*}
\hat{\rho}(1)=\sum_{r,s,\lambda,\lambda', \mu=1}^n D_{r,\lambda,\lambda'} D_{s,\mu,\lambda'}^* \, \vert \psi_n^{(\lambda)} (\phi) \rangle \langle \psi_n^{(\mu)} (\phi) \vert ,\nonumber \\
 \hat{\rho}(2)=\sum_{r,s,\lambda,\lambda', \mu'=1}^n D_{r,\lambda,\lambda'} D_{s,\lambda,\mu'}^* \, \vert \psi_n^{(\lambda')} (\varphi) \rangle \langle \psi_n^{(\mu')} (\varphi) \vert \, ,
\end{eqnarray*}
where $D_{r,\lambda,\lambda'}=\frac{ \mu_n^{(1-r)(\lambda+\lambda'-2)}}{n^2 \mathcal{N}_\lambda \mathcal{N}_{\lambda'}} c_r$. After this, one can calculate the entanglement between the modes. The entanglement is calculated by the linear entropy of the partial density matrices, giving the following result
\begin{equation}
S_L(1)=1- \sum_{\lambda,\mu=1}^n \vert F_{\lambda,\mu}\vert^2\, , \quad F_{\lambda,\mu}=\sum_{r,s,\lambda'=1}^n D_{r,\lambda,\lambda'} D_{s,\mu,\lambda'}^* \, .
\end{equation}
The quantification of the entanglement by using the decomposition of the two-mode system in terms of cyclic states was done in a easier way than by directly taking the expression of the state $\vert T \rangle$ of Eq.~(\ref{tstate}). Several other quantities can be calculated using this decomposition as the mean values and the covariance matrix of the system.

\section{Generalized dihedral states}
The dihedral group of $n$-th order ($D_n$) is a non-Abelian group which contains all the symmetry operations of the $n$-sided regular polygon. In other words, it contains the rotations of the cyclic group $C_n$ and the inversion operators $\hat{U}_r$; $r=1,\ldots,n$. The inversions in the phase space are defined by a rotation plus the complex conjugation operator $\hat{C}$, i.e., $\hat{U}_r=\hat{C}\hat{R}(\theta_r)$, with $\theta_r=2\pi(r-1)/n$. In order to obtain any state associated to the dihedral group, one must impose the condition for the state to be invariant under both the rotations and inversions contained in $D_n$. Inspired by the cyclic states, one can use a superposition of all the rotations and inversions of a noninvariant state $\vert \phi \rangle$, that is the superposition of the states $\hat{R}(\theta_r)\vert \phi \rangle$ and $\hat{U}_r \vert \phi \rangle$. As we have seen in the sections 3 and 4, the superpositions with probability amplitudes given by the characters of the cyclic group $\chi_n^{(\lambda)} (g_r)$ are orthogonal as they contain different photon numbers. Given these arguments we define a set of $n$ dihedral states, each one corresponding to an irreducible representation of the cyclic subgroup $C_n$, as follows
\begin{mydef}
Let $\vert \phi \rangle=\sum_{m=0}^\infty A_m (\phi) \vert m \rangle$ be a quantum state with at least one mean quadrature component ($\hat{x}=(\hat{a}+\hat{a}^\dagger)/\sqrt{2}$, $\hat{p}=i(\hat{a}^\dagger-\hat{a})/\sqrt{2}$) different from zero, i.e., $\langle \phi \vert \hat{x} \vert \phi \rangle \neq 0$, or $\langle \phi \vert \hat{p} \vert \phi \rangle \neq 0$. The general dihedral state for the irreducible representation $\lambda$ of the subgroup $C_n$ is defined as
\begin{equation}
\left\vert \gamma_n^{(\lambda)} (\phi) \right\rangle = \mathcal{N}_\lambda \sum_{r=1}^n (\chi^{(\lambda)}_n (g_r) \vert \phi_r \rangle+\chi^{*(\lambda)}_n (g_r) \vert \phi^*_r \rangle) \, ,
\label{ccy}
\end{equation}
where $\chi_n^{(\lambda)}(g_r)$ is the character associated to the  element of the group $g_r$ of the cyclic group, $\vert \phi^*_r \rangle=\hat{U}_r \vert \phi \rangle=\sum_{m=0}^\infty A^*_m (\phi) e^{i \theta_r m} \vert m \rangle$ ($\theta_r=2\pi (r-1)/n$), and where 
\[
\mathcal{N}_\lambda^{-2}=\sum_{r,r'=1}^n  (\chi_n^{*(\lambda)}(g_{r'}) \langle \phi_{r'} \vert+\chi_n^{(\lambda)}(g_{r'})\langle \phi^*_{r'} \vert)( \chi_n^{(\lambda)}(g_{r}) \vert \phi_r \rangle+\chi_n^{*(\lambda)}(g_{r})\vert \phi^*_r \rangle) \, .
\]
\label{defi3}
\end{mydef}
We would like to emphasize that it is the first time that an orthogonal set of states have been associated to the dihedral group. These set of states are invariant, up to a phase, under the application of all the dihedral group elements. As the construction of the dihedral states corresponds to the sum of two cyclic states: one with initial state $\vert \phi \rangle=\sum_{m=0}^\infty A_m (\phi) \vert m \rangle$ and the other with the initial state $\vert \phi^* \rangle=\sum_{m=0}^\infty A^*_m (\phi) \vert m \rangle$, then the invariance under rotations can be implied from the cyclic states invariance (up to a phase) of Eq.~(\ref{cyc_inv})
\[
\hat{R}(\theta_l)\vert \gamma^{(\lambda)}_n (\phi) \rangle=\mu_n^{(1-\lambda)l}\vert \gamma^{(\lambda)}_n (\phi) \rangle \, ,
\]
from this correspondence one can obtain an expression for the inversions acting on the dihedral states $\hat{U}_l \vert \gamma_n^{(\lambda)}\rangle$ ($\hat{U}_l=\hat{C}\hat{R}(\theta_l)$):
\begin{eqnarray*}
\hat{U}_l \vert \gamma_n^{(\lambda)}\rangle&=&\hat{C}\mu_n^{(1-\lambda)l}\vert \gamma^{(\lambda)}_n (\phi) \rangle \, \\
&=&\mu_n^{(\lambda-1)l} \vert \gamma_n^{(\lambda)} \rangle \, ,
\end{eqnarray*}
and thus one can imply that the dihedral state $\vert \gamma_n^{(\lambda)}\rangle$ in Def.~\ref{defi3} is invariant, up to a phase, under all the elements of the dihedral group $D_n$.

As we can see in Def.~\ref{defi3}, the dihedral group can be defined using the sum of a noninvariant state $\vert \phi \rangle$ and its conjugate $\vert \phi^* \rangle$, this implies that the cyclic state $\vert \psi_n^{(\lambda)}\rangle$ is also a dihedral state $\vert \gamma_n^{(\lambda)}\rangle$ when the initial state has only real photon number probability amplitudes $A_m(\phi)\in \mathbb{R}$, implying $\vert \phi \rangle=\vert \phi^* \rangle$. One can also notice that the dihedral states correspond to the erasure map of the state $(\vert \phi \rangle+\vert \phi^*\rangle)/\sqrt{2}$ since, as stated before, the dihedral state correspond to the sum of the cyclic states for $\vert \phi \rangle$ and $\vert \phi^* \rangle$.

As stated before, the sum $ \chi^{(\lambda)}_n (g_r) \vert \phi \rangle+ \chi^{*(\lambda)}_n (g_r) \vert \phi^*\rangle$ used to obtain the dihedral superpositions, is an state with real probability amplitudes, as $\vert \phi \rangle=\sum_{m=0}^\infty A_m(\phi) \vert m \rangle$, then $ \chi^{(\lambda)}_n (g_r) \vert \phi \rangle+ \chi^{*(\lambda)}_n (g_r) \vert \phi^* \rangle=2 \sum_{m=0}^\infty {\rm Re}( \chi^{(\lambda)}_n (g_r)\, A_m(\phi)) \vert m \rangle$. It can be seen that an analogous procedure to define dihedral states can be done by using the imaginary part of the probability amplitudes $ \chi^{(\lambda)}_n (g_r) A_m(\phi)$, e.g., by using the subtraction of the states $ \chi^{(\lambda)}_n (g_r) \vert \phi \rangle- \chi^{*(\lambda)}_n (g_r)\vert \phi^* \rangle$ instead of the sum $ \chi^{*(\lambda)}_n (g_r)\vert \phi \rangle+ \chi^{*(\lambda)}_n (g_r)\vert \phi^*\rangle$. The states associated to the subtraction are also invariant, up to a phase, under all the transformations contained in the dihedral group, however they are not orthogonal to the states defined in Def.~\ref{defi3}. However, they still can be helpful as they contain the dihedral symmetry.

In fig.~\ref{wignerd3}, the Wigner functions and their contour plots for each one of the three states associated to the dihedral group $D_3$ are shown. To construct this figure, the Gaussian state of Eq.~(\ref{gaus}) with $a=1$ and $b=1+i$ was used to generate the states of $D_3$. In all the cases one can notice that additionally to the rotational symmetry of the $C_3$ subgroup, the inversion invariance is also present.

%Figure 4
\begin{figure}
\centering
\includegraphics[scale=0.28]{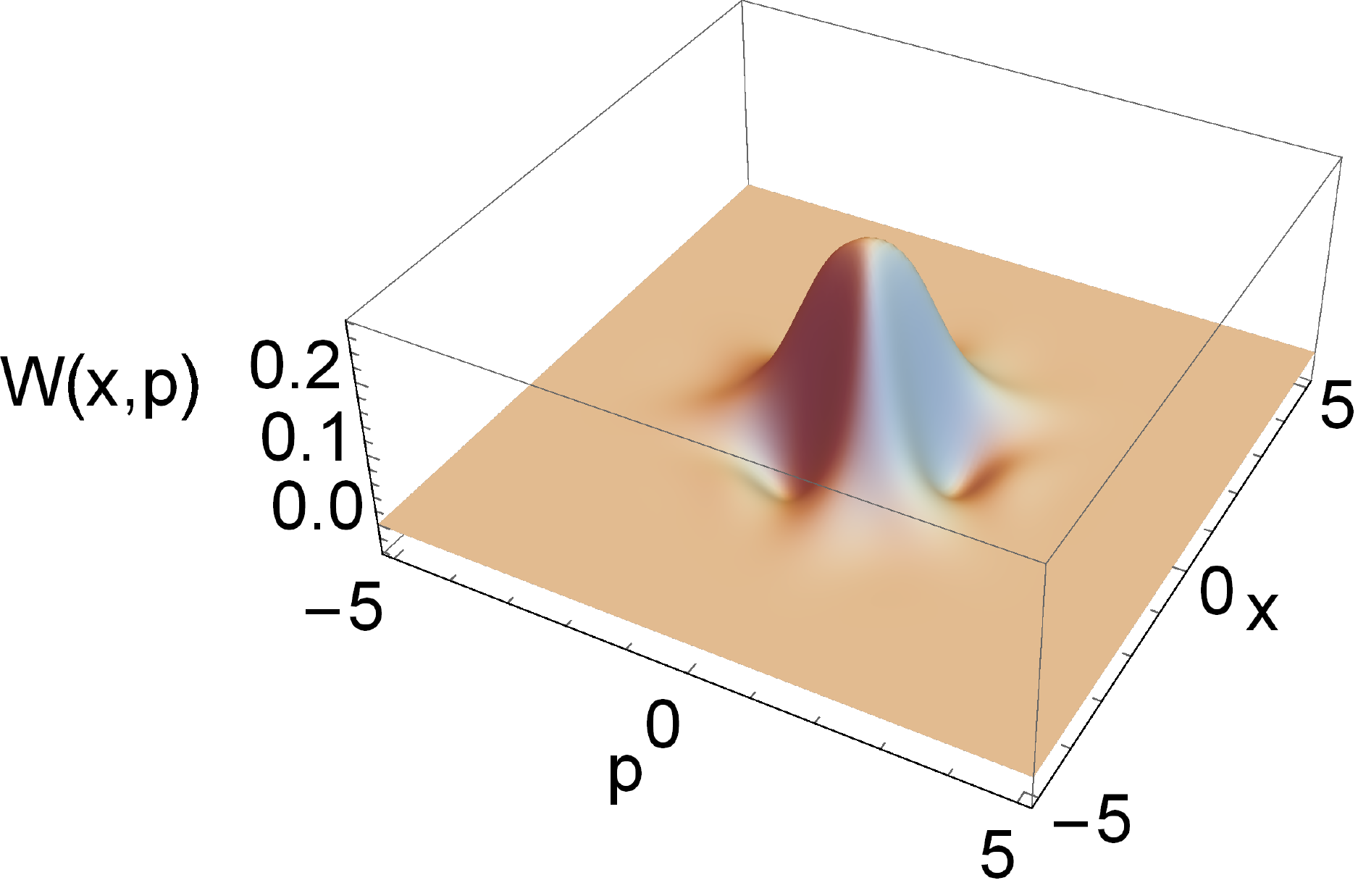} 
\includegraphics[scale=0.28]{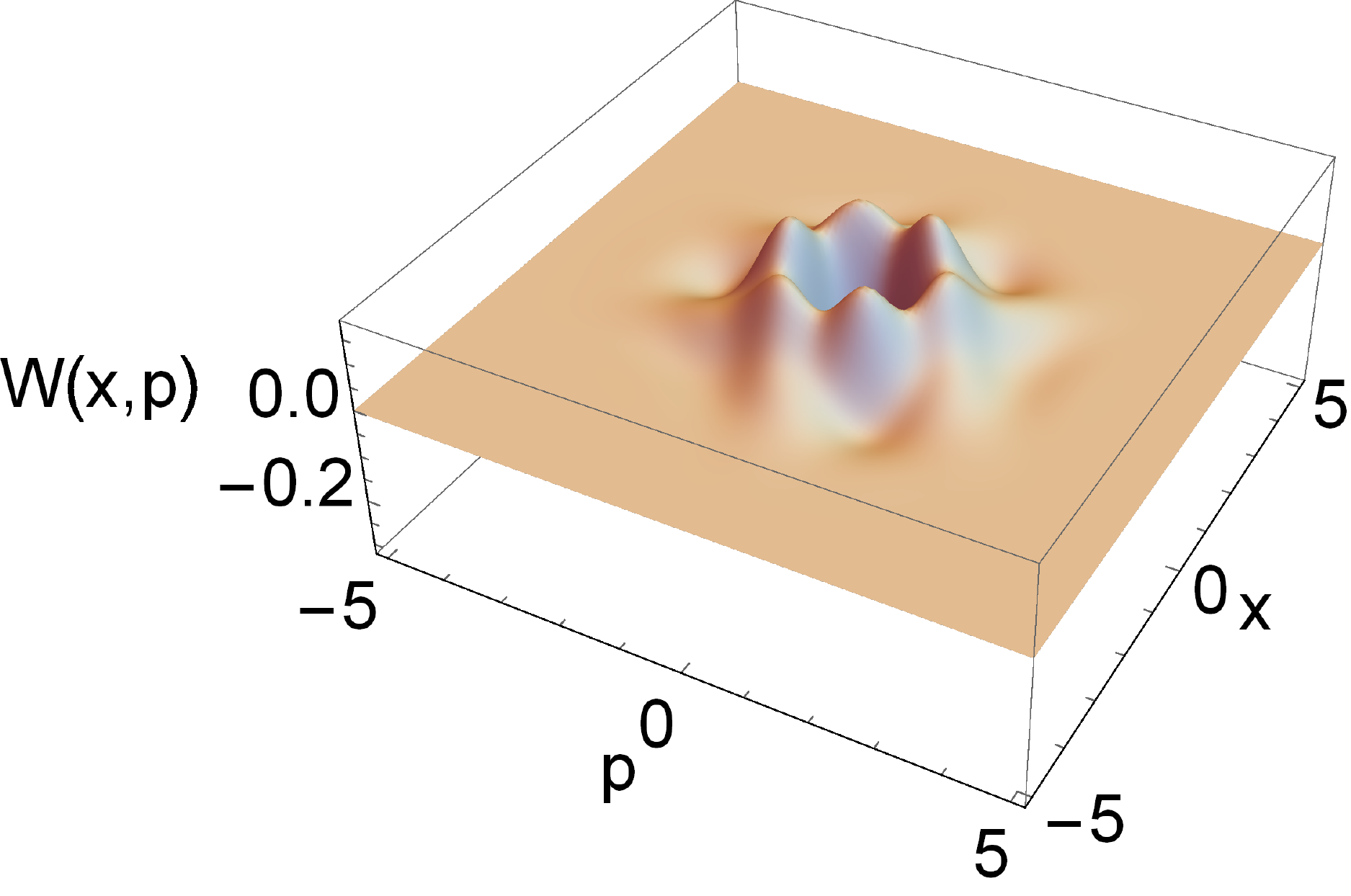} 
\includegraphics[scale=0.28]{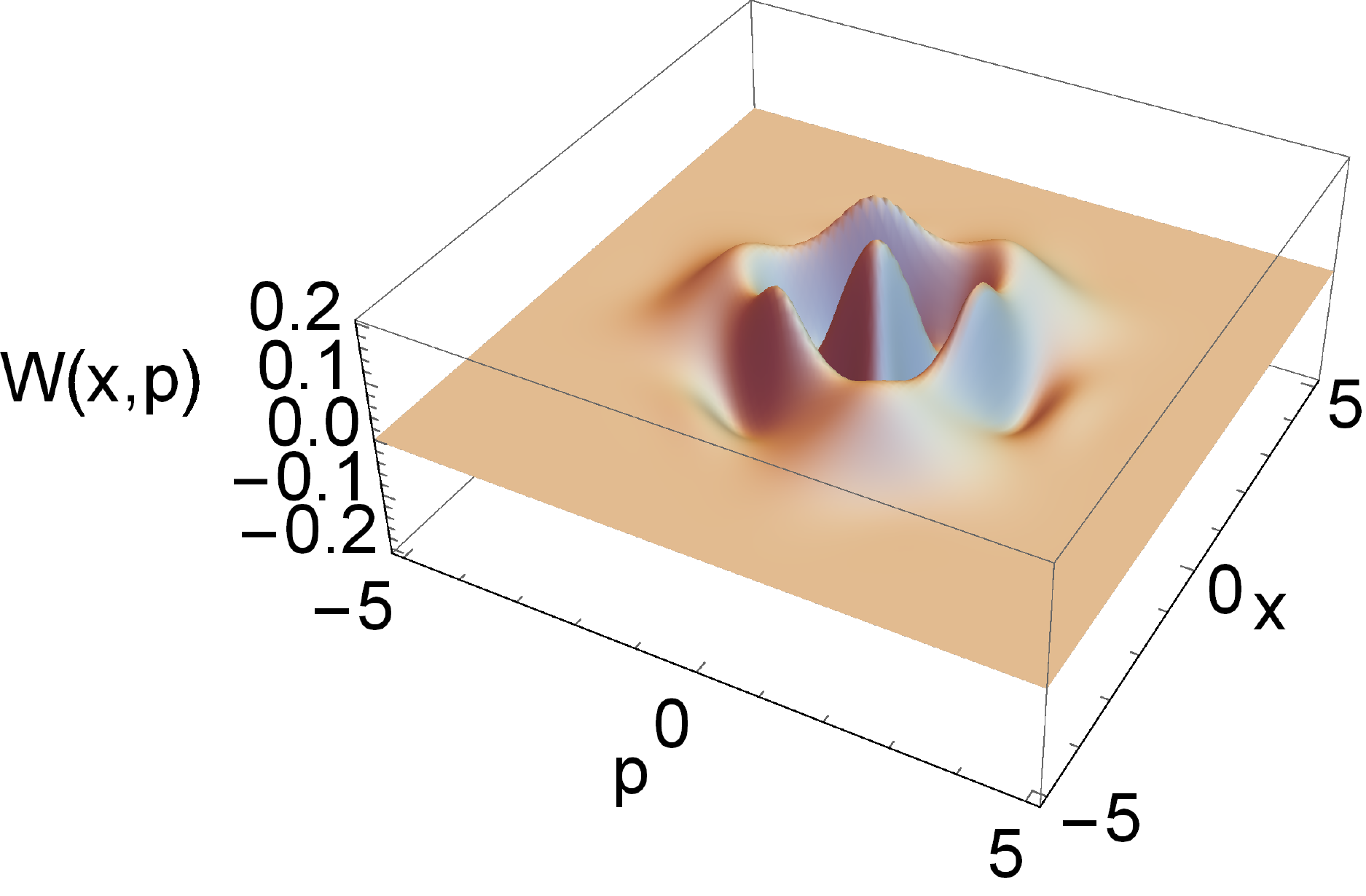} 
\includegraphics[scale=0.38]{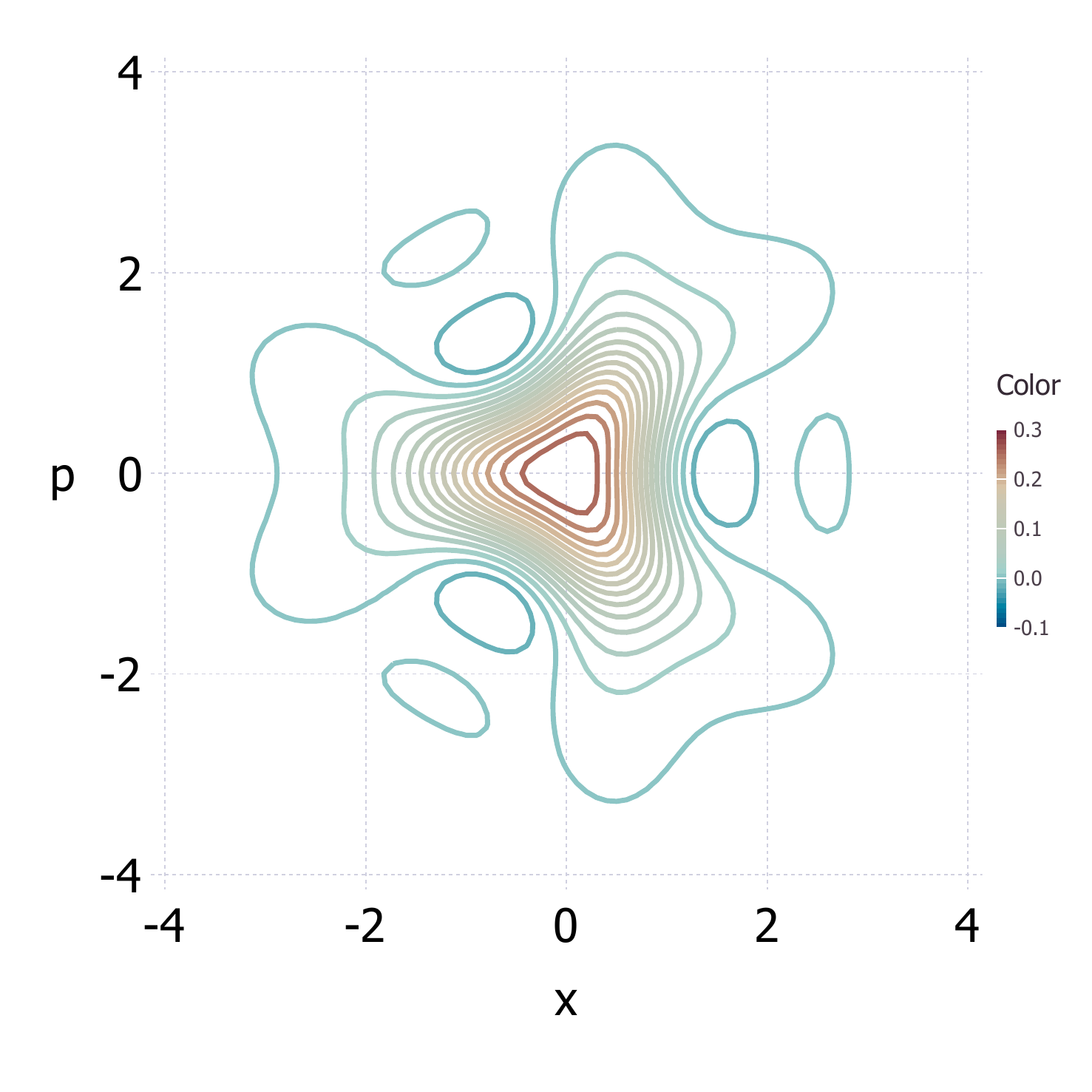}
\includegraphics[scale=0.38]{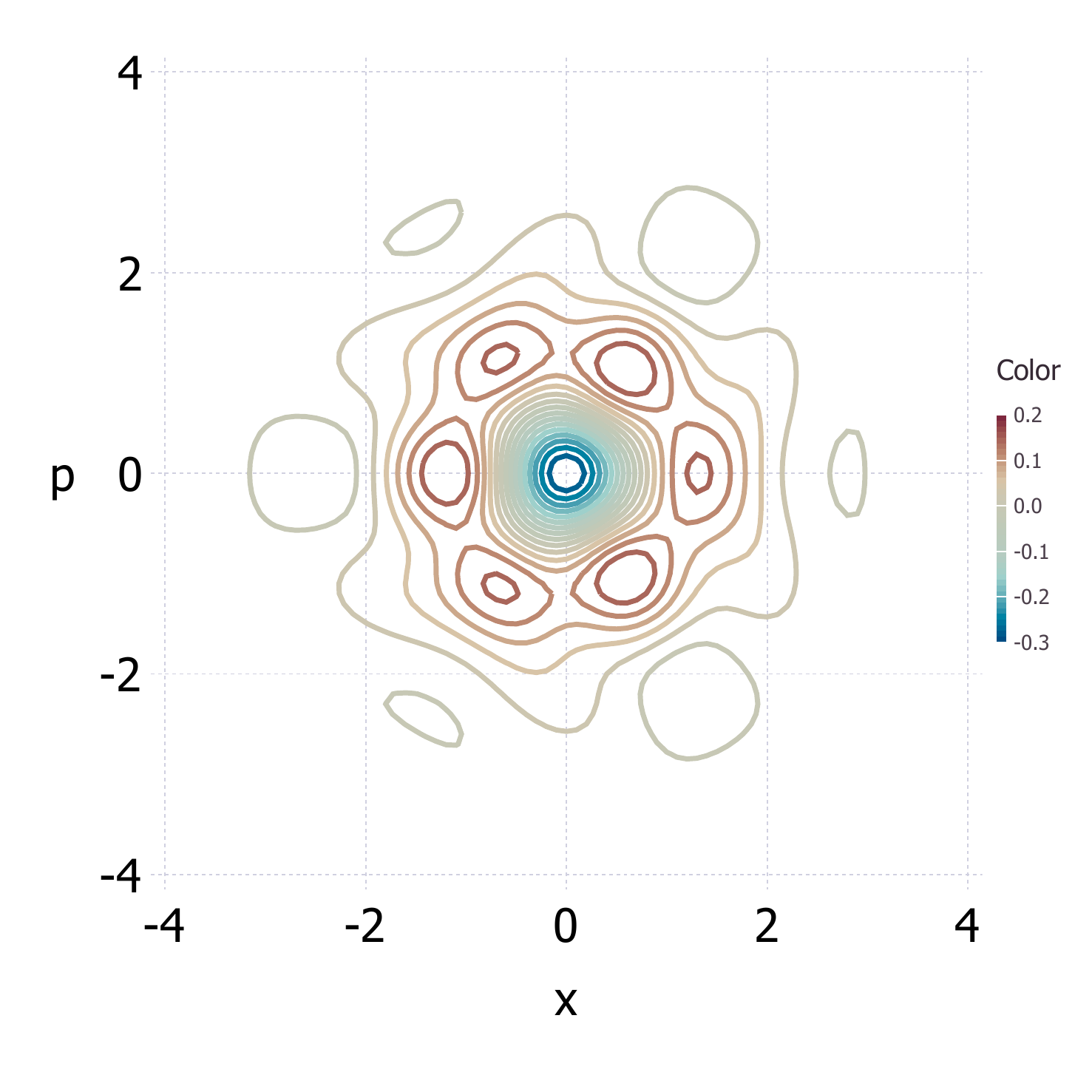}
\includegraphics[scale=0.38]{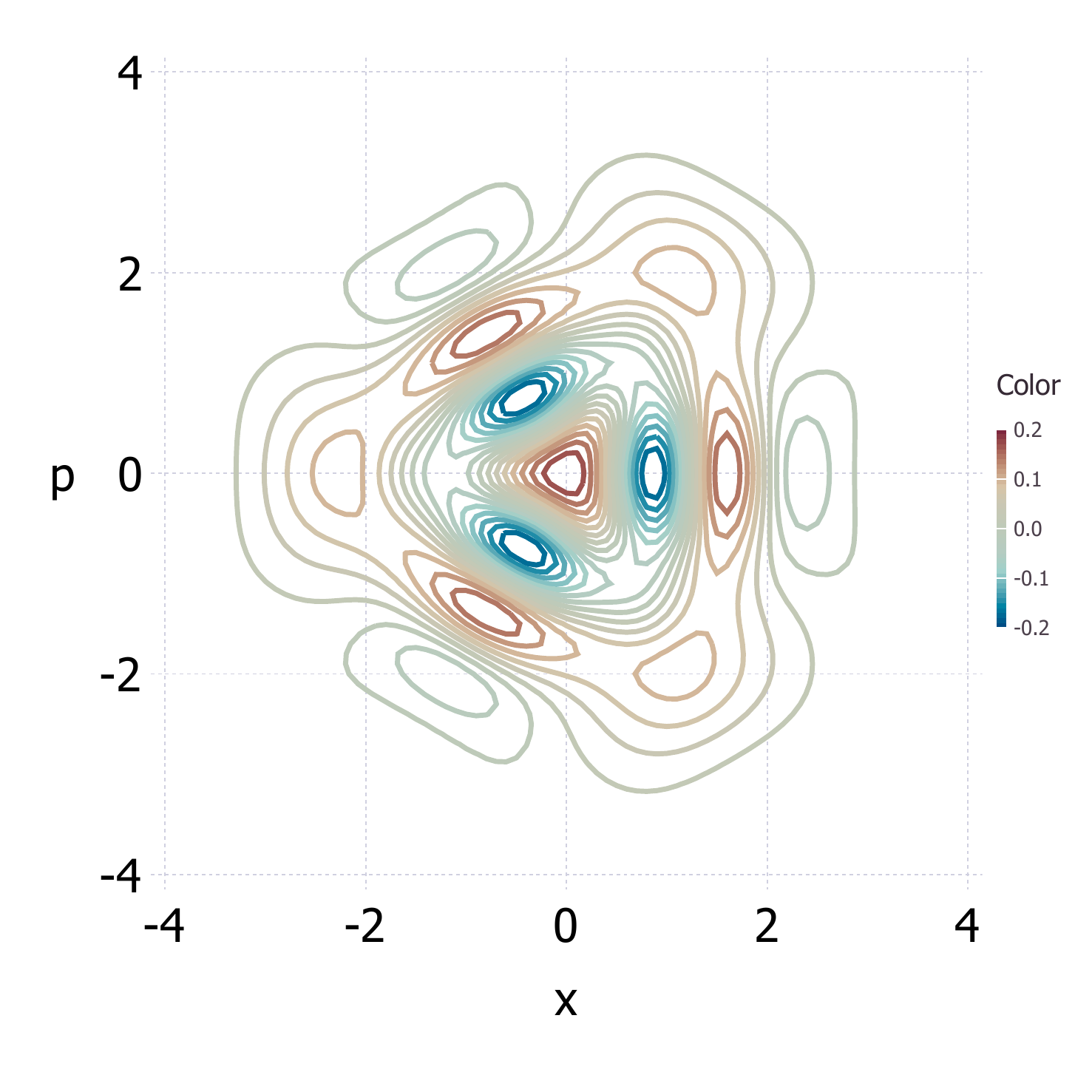}
\caption{Wigner functions and their contour plots for the dihedral Gaussian states associated to the different irreducible representations $\lambda$ of the group $D_3$ for $\lambda=1$ (left), $\lambda=2$ (center), and $\lambda=3$ (right). For these figures the chosen parameters for the initial Gaussian state with $a=1$ and $b=1+i$. \label{wignerd3}}
\end{figure}

\section*{Summary and conclusions}
A general procedure to obtain a set of $n$ orthogonal pure states (or density matrices) associated to each of the irreducible representations of the cyclic group $C_n$ and dihedral group $D_n$ was proposed. This procedure can be summarized as follows: given any state $\vert \phi \rangle$ which is not invariant under the rotations of the cyclic group, the cyclic states can be obtained  from the weighted superposition of the phase-space rotations of the initial state $\hat{R}(\theta_j)\vert \phi \rangle$ ($j=1,\ldots,n$), where the weights of each rotated state are given by the characters of each irreducible representation. This procedure is then extended to density matrices where the weighed superpositions are made of the elements $\hat{R}(\theta_r)\hat{\rho} \hat{R}^\dagger (\theta_s)$, where $\hat{\rho}$ is the initial noninvariant density matrix. Additionally, it was shown that the resulting states associated to $C_n$ provided by our method are invariant, up to a phase, under any element of the group. The associated states to the dihedral group $D_n$ are defined through the rotations of the original noninvariant state $\vert \phi \rangle$ and its complex conjugate $\vert \phi^* \rangle$. In the case of the dihedral states, it is the first time that an orthogonal set of states have been associated to the dihedral group. 

The correspondence between the cyclic states of $C_n$ and the renormalized states obtained after the erasure of certain photon numbers was established and discussed. In particular, it was shown that the cyclic state corresponds, up to a phase, to the renormalized states with photon number states $\vert m \rangle$ erased, where the erased states do not satisfy the condition ${\rm mod}(\lambda+m-1,n)=0$. In an analogous way, the cyclic density matrices obtained by our method correspond to the renormalized matrices where the photon number operators $\vert m \rangle \langle m' \vert$, which does not satisfy the conditions ${\rm mod}(\lambda-m-1,n)=0$ and ${\rm mod}(\lambda-m'-1,n)=0$, are eliminated. On the other hand, the dihedral states correspond to the sum of the cyclic states defined with the states $\vert \phi \rangle$ and $\vert \phi^* \rangle$, for this reason they correspond to the erasure map of the state $(\vert \phi \rangle+\vert \phi^* \rangle)/\sqrt{2}$.

As example of the procedure the general cyclic Gaussian states were defined. It was shown that these states can present subpoissonian photon number statistics by using the Mandel parameter $M_Q=\langle (\Delta \hat{n})^2 \rangle/ \langle \hat{n} \rangle$. The symmetry properties of the cyclic Gaussian states associated to $C_3$ were also checked using the Wigner function. Also, the correspondence between the circle symmetric states $C_n$ ($n\rightarrow \infty$): $\vert \psi_\infty^{(\lambda)}\rangle$ and the Fock states $\vert \lambda-1 \rangle$ was demonstrated.

Also, as an example of the use of the cyclic states, the calculation of the entanglement between subsystems in a two-mode state was presented. This calculation takes advantage of the orthogonality of the cyclic states to define a finite representation of particular bipartite states.

The possible experimental realization of these states was briefly discussed given the evidence presented in \cite{cordero1,cordero2} for a generation of cyclic states in the atom-field interaction, and in \cite{vlastakis} were these type of superposition can be obtained using a superconducting transmon coupled with a cavity resonator.

\section*{Acknowledgments}
This work was partially supported by DGAPA-UNAM (under project IN101619). 

%\section*{References}

\end{document}